\documentclass[10pt,a4paper,fleqn]{amsart}
\usepackage{amssymb, amsmath, amsthm}
\usepackage{MnSymbol}
\usepackage{libertine}
\usepackage[vvarbb,libertine]{newtxmath}
\usepackage{microtype}
\usepackage[hidelinks]{hyperref}
\usepackage[dvipsnames]{xcolor}
\usepackage[style = apa, backend = biber, alldates = year]{biblatex}
\IfFileExists{all.bib}{\addbibresource{all.bib}}{}

\usepackage{mathtools}
\newlength{\myeqskip}
\AtBeginDocument{%
    \setlength\abovedisplayskip{\myeqskip}%
    \setlength\belowdisplayskip{\myeqskip}%
    \setlength\abovedisplayshortskip{\myeqskip-\baselineskip}%
    \setlength\belowdisplayshortskip{\myeqskip}}

\usepackage{algorithm}
\usepackage[kw]{pseudo}
\pseudoset{
    compact,
    label=\scriptsize\sffamily\color{gray}\arabic*,
    ctfont=\sffamily\color{gray},
    ct-left=\ctfont{\# },
    ct-right=,
    line-height=1.2,
    indent-length=3ex,
}
\usepackage{booktabs}
\usepackage{fontawesome5}
\usepackage{caption}
\usepackage[shortlabels]{enumitem}

\usepackage{import}
\usepackage{tikz-cd}

\definecolor{highlight}{HTML}{F5CBCB}
\definecolor{primary}{HTML}{748DAE}
\definecolor{secondary}{HTML}{90AB8B}
\definecolor{nuance}{HTML}{9ECAD6}
\definecolor{subtle}{HTML}{FFEAEA}

\newcommand{\g}[1]{{\color{gray} #1}}

\newcommand{\C}{\mathbb{C}}

\newcommand{\re}{\operatorname{Re}}
\newcommand{\im}{\operatorname{Im}}
\newcommand{\ord}{\operatorname{ord}}

\newcommand{\I}{\mathbf{i}}
\renewcommand{\st}{\ \middle|\ }

\newcommand{\identity}{\operatorname{id}}
\newcommand{\reorimsymb}{q}
\newcommand{\homeq}{\equiv_\text{homotopy}}

\theoremstyle{definition}
\newtheorem{definition}{Definition}
\theoremstyle{plain}
\newtheorem{theorem}[definition]{Theorem}
\newtheorem{lemma}[definition]{Lemma}
\newtheorem{proposition}[definition]{Proposition}
\theoremstyle{remark}

\usepackage{tikz}
\usetikzlibrary{shapes.geometric}
\usetikzlibrary{arrows.meta}
\usetikzlibrary{calc}
\usetikzlibrary{math}

\begin{document}

\title{Computing braids from approximate data}

\author{Alexandre Guillemot}
\address{Inria, Université Paris--Saclay, Palaiseau, 91120, France}
\email{alexandre.guillemot@inria.fr}

\author{Pierre Lairez}
\address{Inria, Université Paris--Saclay, Palaiseau, 91120, France}
\email{pierre.lairez@inria.fr}

\date{\today}

\begin{abstract}
  We study the theoretical and practical aspects of computing braids described by approximate descriptions of paths in the plane.
  Exact algorithms rely on the lexicographic ordering of the points in the plane, which is unstable under numerical uncertainty.
  Instead, we formalize an input model for approximate data, based on a separation predicate. It applies, for example, to paths obtained by tracking the roots of a parametrized polynomial with complex coefficients,
  thereby connecting certified path tracking outputs to exact braid computation.
\end{abstract}

\maketitle

\section{Introduction}
When $n$ points move continuously in the complex plane without colliding, they define a \emph{geometric braid}.
Up to homotopy, such braids can be described combinatorially by recording the winding of the points, leading to an algebraic structure first formalized by \textcite{Artin_1947}: the \emph{braid group}, given by the following presentation
\[ B_n= \left \langle \sigma_1,\ldots,\sigma_{n-1}\st \sigma_i \sigma_{i+1} \sigma_i = \sigma_{i+1} \sigma_i \sigma_{i+1}, |i-j| > 1 \Rightarrow \sigma_i \sigma_j=\sigma_j\sigma_i \right \rangle. \]
Our main case of interest consists of braids arising from a parametrized polynomial $f \in \C[t][x]$.
The displacement of $t$ along a continuous path $\gamma: [0, 1] \to \C$ induces a continuous motion of the roots of $f(t, -)$ in $\C$.
Provided $\gamma$ stays away from a certain critical set, the discriminant locus, the roots do not collide, so they induce a geometric braid.
Given $f$ and $\gamma$, the goal is to compute the associated combinatorial braid.
This finds applications in the study of the topology of complex varieties: braid monodromy can be used to compute the fundamental group of the complement of a curve in $\C\mathbb{P}^2$ (\cite{Zariski_1929}; \cite{Vankampen_1933}; \cite{MarcoBuzunarizRodriguez_2016}) as well as to compute the action of the monodromy on the homology of a fibration \parencite{Pichonpharabod_2024}.

The topological nature of braids leads to a symbolic-numeric approach: first numerically compute the root motion, then recover the combinatorial braid.
The first step can be handled by \emph{certified homotopy continuation} algorithms (\cite{VanDerHoeven_2011}; \cite{BeltranLeykin_2012,BeltranLeykin_2013}; \cite{MarcoBuzunarizRodriguez_2016}; \cite{Kranich_2016}; \cite{XuBurrYap_2018}; \cite{DuffLee_2024}; \cite{GuillemotLairez_2024}).
They output a set of disjoint tubular neighborhoods around the root paths in $\C \times [0, 1]$, which is enough information to characterize the associated braid.
In what follows, the focus is on the second step: given a set of tubular neighborhoods representing a geometric braid, how to compute the associated combinatorial braid?
To avoid limiting the method to a specific representation of the tubular neighborhoods, whose exact nature differs from one path tracking algorithm to another (piecewise linear, Taylor models, ...), an abstract notion of \emph{approximate data} representing a geometric braid is used.
It reduces to what is essential to capture the braid, in the form of a geometric predicate that can be implemented on a specific representation of the tubular neighborhoods.

\subsubsection*{Previous work}
\textcite{MarcoBuzunarizRodriguez_2016} solve the problem of computing braids given $f$ and $\gamma$, and provide an implementation.
Given a set of tubular neighborhoods obtained through their certified homotopy continuation software SIROCCO, they compute piecewise linear paths in each tube from which they can recover the combinatorial braid.
The present method avoids this linearization step.
Another approach, by \textcite{RodriguezWang_2017}, revolves around solving a certain polynomial system whose solution set contains times at which the roots cross when projected onto the real axis, then determining the sign of the crossings using homotopy continuation.
Still, further work is required to make their method certified and to handle geometric braids with non-generic crossings.

\subsubsection*{Contributions}
We introduce a data structure to represent a geometric braid with $n$ strands, consisting of a geometric predicate \fn{sep}, which gives information on the relative positions of the points while it does not retain their precise location.
More precisely, assume the point motion is given by $n$ paths $F_1, \dotsc, F_n : [0, 1] \to \C$, each associated with the movement of one point.
Given $i, j \in [n]$ distinct and $t \in [0, 1)$, $sep(i, j, t)$ returns a time $t' \in (t, 1]$ and an axis $q$ that can be either real or imaginary, such that $F_i$ do not cross $F_j$ when projected onto $q$ in the time interval $[t, t']$.
This data structure adapts well to situations when the point motion is computed approximately: even without access to the exact position of the points, error bounds can be used to determine if two points stay separated along the real or imaginary direction.
Given a braid described by its \fn{sep} predicate, we present an algorithm to compute its associated combinatorial braid.
This method avoids computing piecewise linear approximations of the strands in the tubular neighborhoods.
This linearization step is automatic when the path tracking backend outputs piecewise linear tubular neighborhoods (\cite{MarcoBuzunarizRodriguez_2016}; \cite{DuffLee_2024}); by contrast, linearization requires additional work when the tubular neighborhoods are given by Taylor models, which appear in recent methods for path tracking (\cite{VanDerHoeven_2011}; \cite{GuillemotLairez_2024}).
In the latter case, each tubular neighborhood is given by higher-order pieces in which linearization may require multiple segments, whereas real or imaginary part separation may be checked on the whole piece.

\section{Braids}

Section~\ref{sec:conf-spac-braids} recalls the definitions of geometric and combinatorial braids.
Section~\ref{sec:computing-braids-an})
explains how to compute combinatorial braids from geometric ones in an exact setting.

\subsection{Configuration spaces and braids}
\label{sec:conf-spac-braids}

For~$n \geq 1$, let~$OC_n$ denote the \emph{ordered configuration space of~$n$ points in~$\mathbb{C}$}:
\[ OC_n = \left\{ (z_1,\dotsc,z_n) \in \mathbb{C}^n \st \forall i, j\in [n], i\neq j \Rightarrow z_i \neq z_j \right\}. \]
The permutation group~$S_n$ acts on~$OC_n$ by permuting the $z_i$:
for a permutation~$\sigma \in S_n$,
\[ \sigma \cdot \left( z_1,\dotsc,z_n \right)  = ( z_{\sigma^{-1}(1)},\dotsc,z_{\sigma^{-1}(n)}). \]
The \emph{(unordered) configuration space of~$n$ points in~$\mathbb{C}$} is the quotient space
$C_n = OC_n / S_n$, endowed with the quotient topology.

A \emph{geometric braid} is the homotopy class of a path $F : [0,1] \to C_n$.
Let us fix~$[n] = \left\{ 1,\dotsc,n \right\} \in C_n$ as a base point.
A geometric braid~$F$ such that~$F(0) = F(1) = [n]$ defines an element of the fundamental group~$\pi_1(C_n, [n])$.
\Textcite{Artin_1947} introduced the \emph{braid group}
\[ B_n= \left \langle \sigma_1,\ldots,\sigma_{n-1}\st \sigma_i \sigma_{i+1} \sigma_i = \sigma_{i+1} \sigma_i \sigma_{i+1}, |i-j| > 1 \Rightarrow \sigma_i \sigma_j=\sigma_j\sigma_i \right \rangle \]
and proved that~$B_n \simeq \pi_1(C_n, [n])$.
We call \emph{combinatorial braids} the elements of~$B_n$, as opposed to geometric braids.
The isomorphism is given by mapping the generator~$\sigma_i$ to the path exchanging~$i$ and~$i+1$, with the right point passing above the left point (Figure~\ref{fig:braid-generator}).
\begin{figure}[tbp]
  \centering
  \begin{tikzpicture} [x=1cm,y=1cm]
    \draw[->, thick, color=gray] (4,0) arc[radius = 1, start angle=0, end angle=170];
    \draw[->, thick, color=gray] (2,0) arc[radius = 1, start angle=180, end angle=350];

    \fill (-2,0) node {\ldots};
    \fill (0,0) circle (3pt) node[below left] {$i-1$};
    \fill (2,0) circle (3pt) node[below left] {$i$};
    \fill (4,0) circle (3pt) node[below left] {$i+1$};
    \fill (6,0) circle (3pt) node[below left] {$i+2$};
    \fill (8,0) node {\ldots};

  \end{tikzpicture}
  \caption{Path in~$C_n$ inducing the combinatorial braid~$\sigma_i$.}
  \label{fig:braid-generator}
\end{figure}
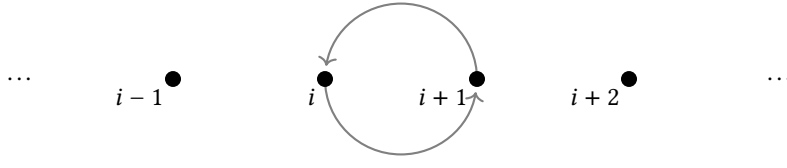

When a geometric braid does not start or end at the base point~$[n]$, it is still possible to associate a combinatorial braid in a natural way.
To this end, we define a canonical path~$R_c$ from~$[n]$ to~$c$ for any configuration~$c = \left\{ z_1,\dotsc,z_n \right\}$ in~$C_n$ as follows: order the points~$z_i$ lexicographically, comparing the real parts first and breaking ties using the imaginary parts.
The path~$R_c$ is given by the linear interpolation
\[ R_c(t) = \left\{ (1-t)i + t z_i \st i\in [n] \right\}. \]
For every~$t \in [0, 1]$, the points are strictly lexicographically ordered, so they never collide.
Given a geometric braid~$F : [0,1] \to C_n$,
we consider the concatenation\footnote{Following the most common convention \parencite[e.g.][]{Hatcher_2002}, the concatenation $f\cdot g$ of a path~$f$ and a path~$g$ traverses~$f$ first, then~$g$.} $\widetilde F = R_{F(0)} \cdot F \cdot R_{F(1)}^{-1}$, which defines an element of~$\pi_1(C_n, [n])$, and therefore induces a combinatorial braid, via Artin's isomorphism.
This convention works well with concatenation:
for any two paths~$F$ and~$G$ in~$C_n$ such that~$F(1) = G(0)$,
the path $\widetilde{F\cdot G}$ is homotopic to~$\widetilde F\cdot \widetilde G$.

\subsection{Computing braids in an exact setting}
\label{sec:computing-braids-an}

To better understand the approximate setting, first consider the problem of computing the combinatorial braid associated to an exact representation of a geometric braid.
There are many possible ways to represent a continuous function~$[0,1] \to C_n$;
for computing braids, a practical one is a piecewise linear representation: pick points~$p_0,\dotsc,p_r \in \mathbb{C}^n$
(or rather~$\mathbb{Q}[\I]^n$, in order to perform exact arithmetic)
and define the function $F : [0,1] \to \mathbb{C}^n$  by piecewise linear interpolation:
\[ F(t) = \left(i + 1 - r t \right) p_i + \left( rt-i \right) p_{i+1}, \quad \tfrac{i}{r} \leq t \leq \tfrac{i+1}{r}. \]
Assume that the $n$ coordinates of this function stay distinct;
then $F : [0,1] \to OC_n$ is a path.
This data structure is \emph{exact}, and it allows easy computation of:
\begin{itemize}
  \item for any~$t\in [0,1]$, the unique permutation $\ord_F(t)$
        such that~$\ord_F(t) \cdot F(t)$ is lexicographically ordered;
  \item the discontinuity points of~$\ord_F$\footnote{It is sufficient to compute a finite superset of the discontinuity points, such as the union over all~$1\leq i<j \leq n$ of the boundary of the set of all~$t \in [0,1]$ such that~$\re(F_i(t)) = \re(F_j(t))$.};
  \item the value of~$\ord_F(t)$ on the right of a discontinuity point.
\end{itemize}

This leads to Algorithm~\ref{algo:folklore}.
In the generic situation where there are no multiple crossings (at any time, at most two points share the same real part), the proof is straightforward because one can always choose~$s=1$ in the decomposition in Line~\ref{line:transpositions}. In non-generic situations, the proof is more technical, and we omit it because it is not the main focus of this article.

\begin{algorithm}[tbp]
  \caption{Braid computation in an exact setting}
  \raggedright
  \begin{pseudo}
    def \fn{braid}(F):\\+
    $T \gets$ \tn{all discontinuity points of~$\ord_F$ (or a superset of those)}\\
    $b \gets 1$ \ct{the identity braid}\\
    $\alpha \gets \ord_F(0)$\\
    for $t\in T$ \tn{in increasing order}:\\+
    \ct{Loop invariant: $\alpha = \lim_{s\to t-} \ord_F(s)$}\\
    $\beta \gets \lim_{s\to t+} \ord_F(s)$ \ct{value on the right of the discontinuity point}\\
    \tn{write $\beta \alpha^{-1}$ as a product of transpositions $(j_s\ j_s+1) \dotsb (j_1 \ j_1+1)$} \label{line:transpositions}\\
    for $i$ from $1$ to $s$:\\+
    $\varepsilon \gets \operatorname{sign}\big(\im\big( F_{\alpha^{-1}(j_i + 1)}(t) - F_{\alpha^{-1}(j_i)}(t) \big) \big)$\\
    $b \gets b \cdot \sigma_{j_i}^\varepsilon $ \ct{update the braid}\\
    $\alpha \gets (j_i \ j_i+1) \cdot \alpha$ \ct{update the strand ordering}\\--
    return $b$
  \end{pseudo}
  \label{algo:folklore}
\end{algorithm}

When the geometric braid~$F$ is specified only approximately (for instance, by tubular neighborhoods),
the ordering function~$\ord_F$ is underspecified and cannot be computed anymore.
One solution is to compute an exact piecewise linear path~$G : [0,1]\to OC_n$ such that:
\begin{enumerate}
  \item the lexicographic ordering of coordinates in~$F(0)$ (resp.~$F(1)$) is identical to that of~$G(0)$ (resp.~$G(1)$);
  \item the path~$F(0) \to G(0) \overset{G}{\to} G(1) \to F(1)$ (where the first and last arrows are straight lines) is homotopic to~$F$.
\end{enumerate}
In this case, the braids induced by~$F$ and~$G$ are the same, and Algorithm~\ref{algo:folklore} applies to~$G$.
This is the approach taken by \textcite{MarcoBuzunarizRodriguez_2016}.
But this still requires a lexicographic ordering of~$F(0)$ and~$F(1)$, which may not be available in a numerical context.
From a practical point of view, even if the ordering of the endpoints is not ambiguous, it may still be hard to compute such a path~$G$, for example when~$F$ is given as tubular neighborhoods of Taylor models, which appear in the context of certified path tracking \parencite{VanDerHoeven_2011,GuillemotLairez_2024}.
The number of line segments required to approximate a nonlinear path within a tubular neighborhood is unbounded as the tube radius decreases.

\section{Approximate data and arrangements}

Section~\ref{sec:coping-with-numer} introduces a definition of \emph{path approximation} that provides the basic context for computing braids from approximate data.
Section~\ref{sec:arrangements} studies \emph{arrangements}, a combinatorial abstraction representing the position of $n$ points in the plane. Sequences of arrangements are used as an intermediate representation of braids, between geometric braids and combinatorial braids.

\subsection{Coping with numerical ambiguity}
\label{sec:coping-with-numer}

The braid induced by a path~$F : [0,1] \to \mathbb{C}^n$ is a topological object, invariant under deformation, so it should be computable from approximate data.
But what does it mean to have approximate data?
Concrete numerical values are not important;
the goal is only to separate the points and capture their winding.
To that end, it is enough to know, for any~$t\in [0,1]$ and any pair~$1\leq i < j \leq n$, a coordinate axis, $\re$ or~$\im$, along which~$F_i(t)$ and~$F_j(t)$ are separated, together with their relative order along this axis.
This leads to the following definition.

\begin{definition}[Path approximation]\label{def:approx}
  An \emph{approximation of a path}~$F:[0,1]\to OC_n$ is the data of two families of open subsets~$(R_{ij})$ and~$(I_{ij})$ of~$[0,1]$, indexed by distinct $i, j\in [n]$,
  such that, for all such pairs,
  \begin{enumerate}[(i)]
    \item $\re(F_i(t)) < \re(F_j(t))$ for all~$t\in R_{ij}$,
    \item $\im(F_i(t)) < \im(F_j(t))$ for all~$t\in I_{ij}$,
          \item\label{it:cover} $R_{ij} \cup R_{ji} \cup I_{ij} \cup I_{ji} = [0,1]$.
  \end{enumerate}
\end{definition}

In computational terms, separation queries are the only access to path approximations: given~$t$ and a pair~$1\leq i < j \leq n$, we want to obtain at least one element of~$\left\{ R_{ij}, R_{ji}, I_{ij}, I_{ji} \right\}$ that contains~$t$, the existence of which is granted by~\ref{it:cover}.
This leads to the following computation-oriented variant of Definition \ref{def:approx}; the rest of the article uses it.
The two definitions are equivalent, in that an approximation for one definition gives an approximation for the other definition.

\begin{definition}[Path approximation, computational variant]\label{def:approx2}
  An
  \emph{approximation of a path}~$F:[0,1]\to OC_n$ is the data, for each distinct $i, j\in [n]$,
  of a computable map
  \[ \fn{sep}_F(i, j, -) : [0,1] \longrightarrow [n]^2 \times \{\re, \im\} \times [0,1] \]
  such that for all~$t\in [0,1]$, the value~$(i', j', q, t') = \fn{sep}_F(i, j, t)$
  satisfies:
  \begin{enumerate}[label={($\textit{sep}_{\alph*}$)}]
    \item $(i', j') = (i, j)$ or~$(j, i)$;\label{point:approx-1}
    \item for all~$s\in [t, t']$,
          $q(F_{i'}(s)) < q(F_{j'}(s))$;\label{point:approx-3}
    \item $t' \geq \min(1, t + \varepsilon_F)$, for some~$\varepsilon_F > 0$
          depending only on~$F$.\label{point:approx-2}
  \end{enumerate}

\end{definition}

\subsection{Arrangements}
\label{sec:arrangements}

An arrangement is a combinatorial structure that represents the relative position of points in the complex plane and reflects the information provided by the approximation of a path (Definitions~\ref{def:approx} and~\ref{def:approx2}).

\begin{definition}[Arrangement, Arrangement cell]
  An \emph{arrangement on $n$ points} is a pair $(\prec_{\re},\prec_{\im})$ of strict partial orders\footnote{A \emph{strict partial order} is a transitive, irreflexive binary relation.} on~$[n]$ such that any distinct $i,j \in [n]$ are comparable in at least one of them.
  The \emph{arrangement cell} attached to an arrangement $(\prec_{\re},\prec_{\im})$ is the subset $A \subset OC_n$ defined by
  \begin{multline*}
    A = \big\{ (z_1,\dotsc,z_n) \in \mathbb{C}^n \bigm| i \prec_{\re} j \Rightarrow \re(z_i) < \re(z_j)\\
    \text{ and } i \prec_{\im} j \Rightarrow \im(z_i) < \im(z_j) \text{ for all } i,j \in [n] \big\}.
  \end{multline*}
\end{definition}

\begin{lemma}\label{lem:convex-open-cover}
  Arrangement cells are nonempty, convex, open subsets of~$OC_n$, and they cover~$OC_n$.
  Moreover, any nonempty intersection of arrangement cells is itself an arrangement cell.
\end{lemma}

\begin{proof}
  Let $(\prec_{\re},\prec_{\im})$ be an arrangement and write $A$ for its cell.
  The inequalities defining $A$ are strict, so~$A$ is open in~$\mathbb{C}^n$, and they are real-linear, so~$A$ is convex.

  To see that $A$ is nonempty, choose linear extensions~$\overline{\prec}_{\re}$ and~$\overline{\prec}_{\im}$ of
  the partial orders $\prec_{\re}$ and $\prec_{\im}$, and choose
  permutations~$\pi$ and~$\phi$ of~$[n]$ such that
  $\pi^{-1}(1) \mathrel{\overline{\prec}_{\re}} \dotsb \mathrel{\overline{\prec}_{\re}} \pi^{-1}(n)$
  and
  $\phi^{-1}(1) \mathrel{\overline{\prec}_{\im}} \dotsb \mathrel{\overline{\prec}_{\im}} \phi^{-1}(n)$.
  The configuration $\bigl(\pi(k) + \phi(k)\I\bigr)_{1\leq k \leq n}$ is in~$A$.

  The inclusion $A \subseteq OC_n$ is immediate: given $z \in A$ and $i\neq j$, the definition of an arrangement
  requires one of~$i \prec_{\re} j$, $j \prec_{\re} i$, $i\prec_{\im} j$, or $j\prec_{\im} i$ to hold, which implies
  either $\re(z_i) \neq \re(z_j)$ or $\im(z_i) \neq \im(z_j)$, so $z_i \neq z_j$.

  For the covering property, let $z\in OC_n$ and define $i \prec_{\re} j$ whenever $\re(z_i) < \re(z_j)$, and $i \prec_{\im} j$ whenever $\im(z_i) < \im(z_j)$.
  Distinct points of~$z$ differ in real or imaginary part, so every pair becomes comparable in at least one order.
  So~$(\prec_{\re},\prec_{\im})$ is an arrangement whose cell, by construction, contains~$z$.

  Finally, consider two arrangements $(\prec_{\re},\prec_{\im})$ and $(\prec'_{\re},\prec'_{\im})$ with cells $A$ and $A'$.
  If some $i$ and $j$ satisfy $i \prec_{\re} j$ and $j \prec'_{\re} i$, then $A \cap A' = \varnothing$, and similarly for~$\prec_{\im}$ and~$\prec_{\im}'$.
  So if~$A \cap A' \neq \varnothing$, then the binary relations
  $\prec_{\re} \cup \prec'_{\re}$ and $\prec_{\im} \cup \prec'_{\im}$ are acyclic, and their transitive closures yield partial orders. The intersection $A \cap A'$ is the arrangement cell attached to this pair of orders.
\end{proof}

\begin{definition}[Cover]\label{def:cover}
  A finite sequence $(A_1,\dotsc,A_r)$ of arrangement cells \emph{covers} a path $F : [0,1] \to OC_n$ if there exists a partition $0 = t_0 \leq t_1 \leq \dotsb \leq t_r = 1$ such that $F([t_{i-1},t_i]) \subseteq A_i$ for every $i$.
\end{definition}

Because $OC_n$ is an open subset of $\mathbb{C}^n$ covered by arrangement cells, compactness of $[0,1]$ guarantees that every path in~$OC_n$ admits a finite covering sequence.

Recall that two paths $F,G : [0,1] \to OC_n$ are \emph{homotopic} if there exists a continuous map $H : [0,1]^2 \to OC_n$ with $H(0,t) = F(t)$, $H(1,t) = G(t)$, and $H(s,0) = F(0) = G(0)$, $H(s,1) = F(1) = G(1)$ for all~$s,t$.
The next lemma states that the homotopy class of a path is determined by its endpoints together with any arrangement sequence covering it.

\begin{lemma}\label{lem:homotopy-arrangement}
  Let $F, G : [0,1] \to OC_n$ be two paths.
  Suppose
  \begin{enumerate}[(i)]
    \item $F(0) = G(0)$ and $F(1) = G(1)$, and
    \item there exists a sequence of arrangement cells that covers both~$F$ and~$G$.
  \end{enumerate}
  Then $F$ and $G$ are homotopic.
\end{lemma}

\begin{proof}
  Let $(A_1,\dotsc,A_r)$ be a common covering sequence and choose partitions $0 = t_0 \leq \dotsb \leq t_r = 1$ for $F$ and $0 = s_0 \leq \dotsb \leq s_r = 1$ for $G$ realizing the cover.
  Up to homotopy, one can reparametrize~$F$ and $G$ so that~$s_i = t_i$ and inequalities are strict.

  Define $H : [0,1]^2 \to \mathbb{C}^n$ by $H(s,t) = (1-s)F(t) + s G(t)$.
  This is a homotopy between~$F$ and~$G$:
  the boundary conditions are clearly satisfied, so it suffices to check that $H(s,t)$ stays inside~$OC_n$.
  Fix $t \in [0,1]$ and choose $i$ with $t_{i-1} \leq t \leq t_i$.
  By construction, both $F(t)$ and $G(t)$ lie in $A_i$, and $A_i$ is convex; hence $H(s,t) \in A_i \subset OC_n$ for all $s$.
\end{proof}

\begin{algorithm}[tbp]
  \caption{Computation of an arrangement sequence covering a path}
  \raggedright
  \begin{description}
    \item[input] An approximation of a path~$F : [0,1] \to OC_n$, given as a \emph{sep} method.
    \item[output] A sequence of arrangements covering~$F$.
  \end{description}

  \begin{pseudo}
    %\begin{noident}
    def \fn{cover}(F):\\+
    $A \gets \{\}$ \ct{edge set of a graph encoding the current arrangement}\\
    $\id{lifetime} \gets \left\{  \right\}$ \ct{map edges of~$A$ to their lifetime}\\
    $\mathcal{L} \gets []$ \ct{list of arrangements covering $F$}\\
    $t \gets 0$\\
    while $t < 1$:\\+
    \ct{Complete the graph until it represents an arrangement containing $F(t)$.}\\
    while \tn{$A$ does not define an arrangement}:\label{line:braid:arrangement-repair}\\+
    $i, j \gets$ \tn{a non-comparable pair}\\
    $e, t' \gets \fn{sep}(F, i, j, t)$ \ct{$e \in [n]^2 \times \{\re, \im\}$}\\
    $A \gets A \cup \{ e \}$ \ct{$i$ and~$j$ are now comparable}\\
    $\id{lifetime}[e] \gets t'$\\-
    \tn{append a copy of $A$ to $\mathcal{L}$}\label{line:braid:arrangement-copy}\\

    \ct{Advance time until some edge expires.}\\
    $e \gets$ \tn{edge of~$A$ with smallest lifetime}\\
    $A \gets A \setminus \{ e \}$\\
    $t \gets \id{lifetime}[e]$\\-
    return $\mathcal{L}$
    %\end{noident}
  \end{pseudo}
  \label{algo:cover}
\end{algorithm}

\begin{figure}
  \def\motion#1#2{
    \tikzmath{integer \iter; \iter=#2/18;}
    \draw[sync] (#1 |- top) node[above] {$t = \tfrac{\iter}{5}$} -- (#1 |- bot);
    \begin{scope}[shift=(#1 |- fig), label distance=-2, radius = 2]
      \path[arr] (2,0)  arc[start angle=0, end angle=#2]  node[label={#2:\small 1}] (A) {};
      \path[arr] (0,2)  arc[start angle=90, end angle=90 + #2]  node[inner sep=3pt,  label={90+#2:\small 2}] (B) {};
      \path[arr] (-2,0)  arc[start angle=180, end angle=180+#2]  node[label={180+#2:\small 3}] (C) {};
      \path[arr] (0,-2)  arc[start angle=270, end angle=270+#2]  node[label={270+#2:\small 4}] (D) {};

      \node[dot] at (A) {};
      \node[dot, inner sep=3pt] at (B) {};
      \node[dot] at (C) {};
      \node[dot] at (D) {};
      % \ifnum #2 < 90
      %   \draw[] (A) arc[radius = 2, start angle=#2, end angle=90];
      %   \draw[] (B) arc[radius = 2, start angle=90+#2, end angle=180];
      %   \draw[] (C) arc[radius = 2, start angle=180+#2, end angle=270];
      %   \draw[] (D) arc[radius = 2, start angle=270+#2, end angle=360];
      % \fi
    \end{scope}
  }

  \def\eps{0.4}
  \def\rect#1#2#3#4#5{
    \begin{scope}[xshift=-.8cm, yshift=1ex]
      \fill[rect] (#1 - \eps, #2 - \eps)
      -- (#1 - \eps, #4 + \eps) -- (#3 + \eps, #4 + \eps)
      -- (#3 + \eps, #2 - \eps) -- (#1 - \eps, #2 - \eps);
      \node at ($0.5*(#1+#3, #2+#4)$) {}; % {\color{white}\scriptsize\bfseries #5};
    \end{scope}
  }

  \def\rlabel#1#2#3{
    \begin{scope}[xshift=-.8cm, yshift=1ex, every node/.style={}]
      \node[rlabel] at (#1, #2) {\color{white}\scriptsize\bfseries #3};
    \end{scope}
  }

  \def\perm#1#2#3#4#5#6#7#8{
    \begin{scope}[xshift=-.8cm, yshift=1ex, every node/.style={circle, inner sep=0.05cm, draw=red, thick}]
      \node[label=1] at (#1, #2) {};
      \node[label=2] at (#3, #4) {};
      \node[label=3] at (#5, #6) {};
      \node[label=4] at (#7, #8) {};
    \end{scope}
  }

  \centering
  \makebox[\textwidth][c]{
    \begin{tikzpicture}[x=.3cm, y=.3cm, scale=0.8,
        dot/.style={fill=primary, opacity=1.0, circle, inner sep=1.5pt},
        arr/.style={draw, thin, densely dotted},
        bar/.style={line width=4, draw=nuance},
        sync/.style={dotted},
        rect/.style={primary},
        highlight/.style={circle, inner sep=1pt, fill=highlight, opacity=0.8, text opacity=1},
        rlabel/.style={},
      ]

      \node (label) at (-6, 0) {};

      \def\onedeg{0.6}

      \foreach \s in {0,9,...,90} {
          \node (A\s) at (\s * \onedeg, 0) {};
        }

      \node (top) at (0, 3) {};
      \node (12re) at (0, -0) {};
      \node (12im) at (0, -1) {};
      \node (13re) at (0, -3) {};
      \node (13im) at (0, -4) {};
      \node (14re) at (0, -6) {};
      \node (14im) at (0, -7) {};
      \node (23re) at (0, -9) {};
      \node (23im) at (0, -10) {};
      \node (24re) at (0, -12) {};
      \node (24im) at (0, -13) {};
      \node (34re) at (0, -15) {};
      \node (34im) at (0, -16) {};
      \node (bot) at (0, -16) {};
      \node (fig) at (0, 12) {};
      \node (arrangements) at (0, -22) {};

      \draw[bar] (A0 |- 12re) node[above right] {\small $\im(z_1) < \im(z_2)$} -- ++(17*\onedeg, 0);
      \draw[bar] (A0 |- 12re) +(18*\onedeg, 0) node[above right] {\small $\re(z1) > \re(z_2)$} -- +(90*\onedeg, 0);

      \draw[bar] (A0 |- 13re) node[above right] {\small $\re(z_1) > \re(z_3)$} -- ++(71*\onedeg, 0);
      \draw[bar] (A0 |- 13re) + (72*\onedeg, 0) node[above right] {\small $\im(z_1) > \im(z_3)$} -- +(90*\onedeg, 0);

      \draw[bar] (A0 |- 14re) node[above right] {\small $\im(z_1) > \im(z_4)$} -- +(90*\onedeg, 0);

      \draw[bar] (A0 |- 23re) node[above right] {\small $\im(z_2) > \im(z_3)$} -- +(90*\onedeg, 0);

      \draw[bar] (A0 |- 24re) node[above right] {\small $\im(z_2) > \im(z_4)$} -- ++(53*\onedeg, 0);
      \draw[bar] (A0 |- 24re) + (54*\onedeg, 0) node[above right] {\small $\re(z_2) < \re(z_4)$} -- +(90*\onedeg, 0);

      \draw[bar] (A0 |- 34re) node[above right] {\small $\im(z_3) > \im(z_4)$} -- ++(35*\onedeg, 0);
      \draw[bar] (A0 |- 34re) + (36*\onedeg, 0) node[above right] {\small $\re(z_3) < \re(z_4)$} -- +(90*\onedeg, 0);

      \node at (label |- top) {$i$\quad $j$};
      \node at (label |- 12re) {1\quad 2};
      \node at (label |- 13re) {1\quad 3};
      \node at (label |- 14re) {1\quad 4};
      \node at (label |- 23re) {2\quad 3};
      \node at (label |- 24re) {2\quad 4};
      \node at (label |- 34re) {3\quad 4};

      \motion{A0}{0}
      \motion{A18}{18}
      \motion{A36}{36}
      \motion{A54}{54}
      \motion{A72}{72}
      \motion{A90}{90}

      \begin{scope}[shift=(A0 |- arrangements)]
        \rect{1}{1}{4}{1}{4}
        \rect{1}{4}{4}{4}{2}
        \rect{3}{2}{4}{3}{1}
        \rect{1}{2}{2}{3}{3}

        \rlabel{4}{2}{1}
        \rlabel{2}{4}{2}
        \rlabel{1}{3}{3}
        \rlabel{3}{1}{4}
      \end{scope}

      % \begin{scope}[shift=(A9 |- arrangements)]
      %   \rect{1}{1}{4}{1}{4}
      %   \rect{1}{4}{3}{4}{2}
      %   \rect{4}{2}{4}{3}{1}
      %   \rect{1}{2}{2}{3}{3}

      %   \begin{scope}[rlabel/.style={highlight}]
      %     \rlabel{4}{2}{1}
      %     \rlabel{2}{4}{2}
      %     \rlabel{1}{3}{3}
      %     \rlabel{3}{1}{4}
      %   \end{scope}
      % \end{scope}

      \begin{scope}[shift=(A18 |- arrangements)]
        \rect{1}{1}{4}{1}{4}
        \rect{1}{3}{3}{4}{2}
        \rect{4}{2}{4}{4}{1}
        \rect{1}{2}{3}{2}{3}
        \rlabel{4}{3}{1}
        \rlabel{2}{4}{2}
        \rlabel{1}{2}{3}
        \rlabel{3}{1}{4}

      \end{scope}

      % \begin{scope}[shift=(A27 |- arrangements)]
      %   \rect{3}{1}{4}{1}{4}
      %   \rect{1}{3}{3}{4}{2}
      %   \rect{4}{2}{4}{4}{1}
      %   \rect{1}{2}{2}{2}{3}

      %   \begin{scope}[rlabel/.style={highlight}]
      %     \rlabel{4}{3}{1}
      %     \rlabel{2}{4}{2}
      %     \rlabel{1}{2}{3}
      %     \rlabel{3}{1}{4}
      %   \end{scope}
      % \end{scope}

      \begin{scope}[shift=(A36 |- arrangements)]
        \rect{3}{1}{4}{1}{4}
        \rect{1}{3}{3}{4}{2}
        \rect{4}{2}{4}{4}{1}
        \rect{1}{1}{2}{2}{3}

        \rlabel{4}{3}{1}
        \rlabel{2}{4}{2}
        \rlabel{1}{2}{3}
        \rlabel{3}{1}{4}
      \end{scope}

      % \begin{scope}[shift=(A45 |- arrangements)]
      %   \rect{3}{1}{4}{1}{4}
      %   \rect{1}{3}{2}{4}{2}
      %   \rect{3}{2}{4}{4}{1}
      %   \rect{1}{1}{2}{2}{3}

      %   \begin{scope}[rlabel/.style={highlight}]
      %     \rlabel{4}{3}{1}
      %     \rlabel{2}{4}{2}
      %     \rlabel{1}{2}{3}
      %     \rlabel{3}{1}{4}
      %   \end{scope}
      % \end{scope}

      \begin{scope}[shift=(A54 |- arrangements)]
        \rect{3}{1}{4}{2}{4}
        \rect{1}{3}{2}{4}{2}
        \rect{3}{3}{4}{4}{1}
        \rect{1}{1}{2}{2}{3}

        \rlabel{4}{3}{1}
        \rlabel{2}{4}{2}
        \rlabel{1}{2}{3}
        \rlabel{3}{1}{4}

        \begin{scope}[xshift=-.8cm, yshift=1ex, every node/.style={}]
          \draw[{Classical TikZ Rightarrow[length=2]}-{Classical TikZ Rightarrow[length=2]}, very thick, white] (3,2)--(3,3);
        \end{scope}
      \end{scope}
      \path (A54 |- arrangements) -- ++(2,5) node {*};

      % \begin{scope}[shift=(A63 |- arrangements)]
      %   \rect{3}{1}{4}{3}{4}
      %   \rect{1}{3}{2}{4}{2}
      %   \rect{3}{4}{4}{4}{1}
      %   \rect{1}{1}{2}{2}{3}

      %   \begin{scope}[rlabel/.style={highlight}]
      %     \rlabel{4}{4}{1}
      %     \rlabel{1}{3}{2}
      %     \rlabel{2}{2}{3}
      %     \rlabel{3}{1}{4}
      %   \end{scope}
      % \end{scope}

      \begin{scope}[shift=(A72 |- arrangements)]
        \rect{3}{1}{4}{3}{4}
        \rect{1}{3}{1}{4}{2}
        \rect{2}{4}{4}{4}{1}
        \rect{1}{1}{2}{2}{3}

        \rlabel{4}{4}{1}
        \rlabel{1}{3}{2}
        \rlabel{2}{2}{3}
        \rlabel{3}{1}{4}
      \end{scope}

      % \begin{scope}[shift=(A81 |- arrangements)]
      %   \begin{scope}[yshift=-.7cm]
      %     \rect{1}{1}{2}{1}{3}
      %     \rect{2}{4}{4}{4}{1}
      %     \rect{3}{2}{4}{3}{4}
      %     \rect{1}{2}{1}{3}{2}

      %     \begin{scope}[rlabel/.style={highlight}]
      %       \rlabel{4}{4}{1}
      %       \rlabel{1}{3}{2}
      %       \rlabel{2}{1}{3}
      %       \rlabel{3}{2}{4}
      %     \end{scope}
      %   \end{scope}
      % \end{scope}

      % \begin{scope}[shift=(A90 |- arrangements)]
      %   \begin{scope}[yshift=-.7cm]
      %     \rect{1}{1}{4}{1}{3}
      %     \rect{1}{4}{4}{4}{1}
      %     \rect{3}{2}{4}{3}{4}
      %     \rect{1}{2}{2}{3}{2}

      %     \begin{scope}[rlabel/.style={highlight}]
      %       \rlabel{4}{2}{4}
      %       \rlabel{2}{4}{1}
      %       \rlabel{1}{3}{2}
      %       \rlabel{3}{1}{3}
      %     \end{scope}
      %   \end{scope}
      % \end{scope}

    \end{tikzpicture}}
  \caption{Illustration of Algorithm~\ref{algo:cover} for four points moving along a circle, with position known only approximately (as represented by the blue discs). The bars represent the lifetime of an edge in the graph representing the current arrangement. The diagrams in the last row represent the resulting cover.
    The arrangement cell marked with~$*$ is not representable as a configuration of boxes, it is defined by \[ \left\{ \max(\re(z_2), \re(z_3)) < \min(\re(z_1), \re(z_4)), \im(z_3) < \im(z_2) \text{ and } \im(z_4) < \im(z_1) \right\}. \]}
  \label{fig:example-sep}
\end{figure}
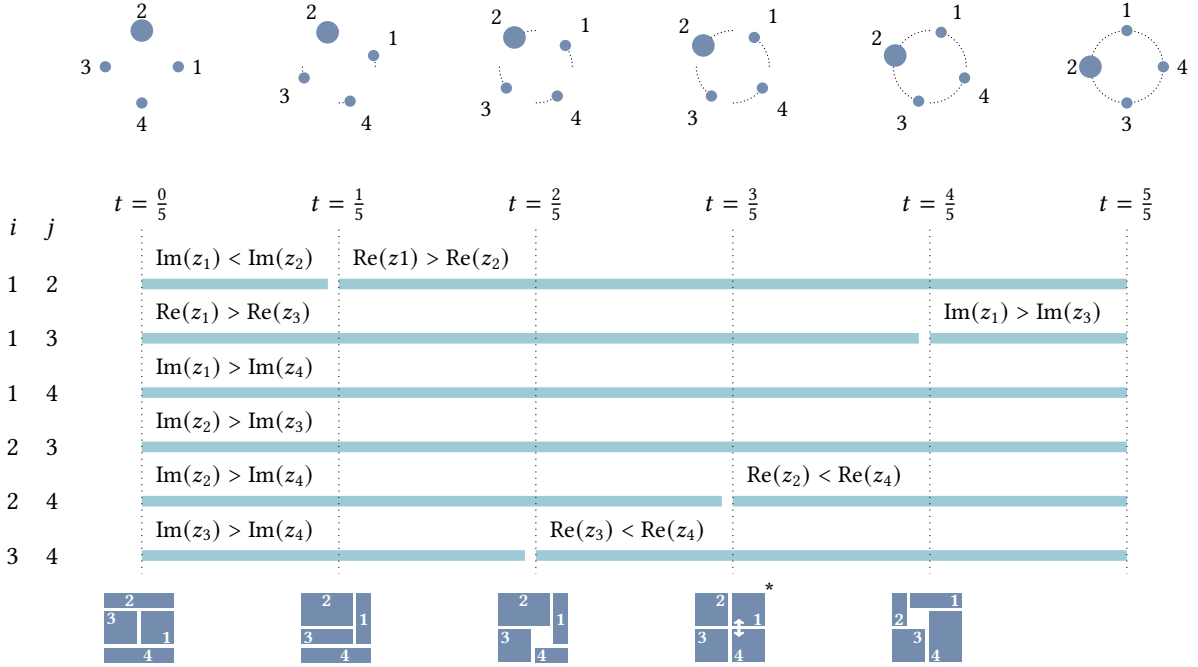

Algorithm~\ref{algo:cover} shows how to implement the computation of a covering sequence of arrangements from a path approximation.
It is exemplified in Figure~\ref{fig:example-sep}. In this example, the algorithm produces a sequence of 5~arrangements.
Arrangements are represented by subsets of~$[n]^2\times \left\{ \re, \im \right\}$, that we interpret as graphs with vertex set~$[n]$ and edges labeled by~$\re$ or~$\im$. The partial orders~$\prec_{\re}$ and~$\prec_{\im}$ are the reachability relations, restricted to $\re$- or $\im$-labeled edges.
The map \id{lifetime} stores the time until each edge is valid.

\begin{proposition}
  Algorithm~\ref{algo:cover} terminates and outputs a sequence covering~$F$.
\end{proposition}

\begin{proof}
  For termination, let $\varepsilon > 0$ be given by Item~\ref{point:approx-2} and, for each integer $k \geq 0$, let $n_k$ denote the number of edges in~$A$ whose lifetime lies in $[k\varepsilon,(k+1)\varepsilon)$.
  Considering the tuple $(n_k)_{0 \leq k \leq \frac{1}{\epsilon}}$ at the end of each iteration, we obtain a strictly decreasing sequence for the lexicographic ordering.
  Thus only finitely many iterations occur.

  Correctness rests on the following invariant coming from the defining properties of~\fn{sep}: whenever~$(i, j, q)$ belongs to~$A$ with lifetime~$t'$, the inequality $q(F_i(s)) < q(F_j(s))$ holds for all $s \in [t,t']$.
\end{proof}

\begin{figure}[tbp]
  \centering
  \includegraphics[width=.08\linewidth]{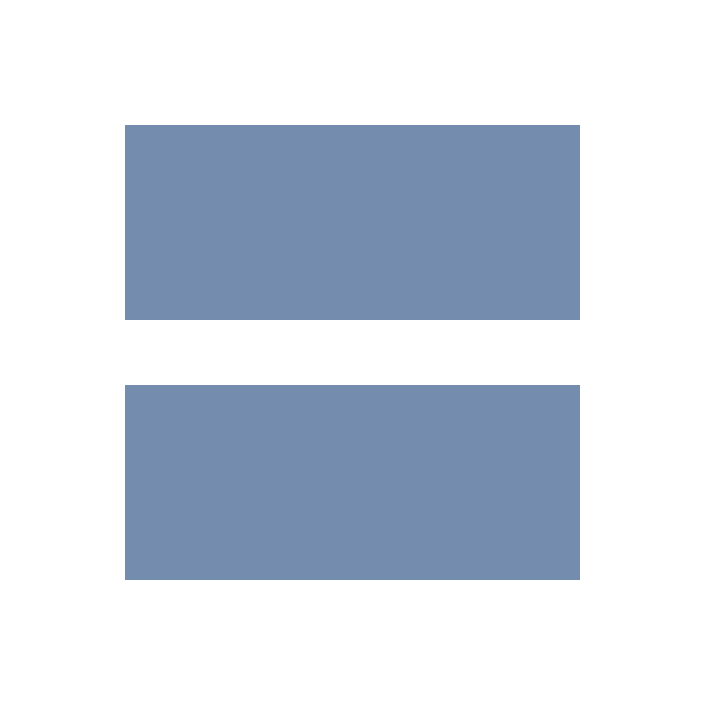}
  \includegraphics[width=.08\linewidth]{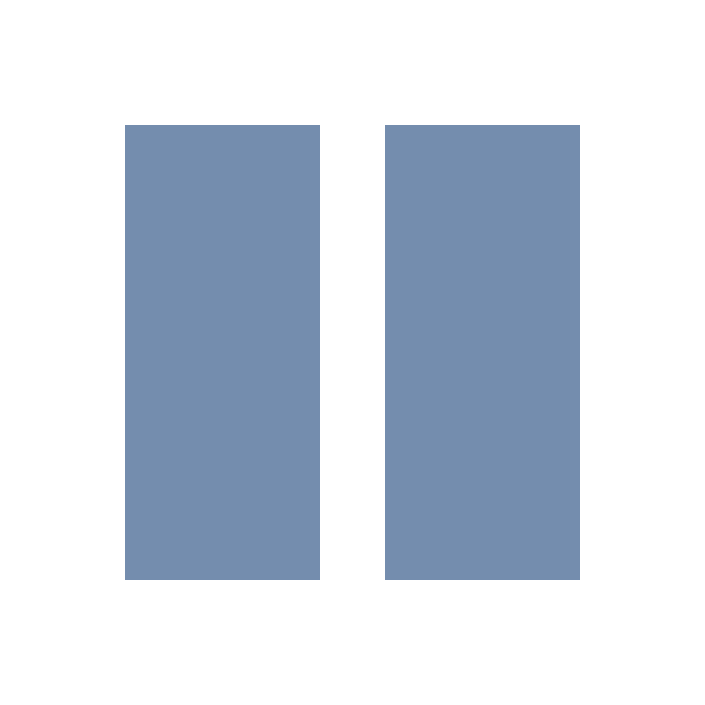}
  \includegraphics[width=.08\linewidth]{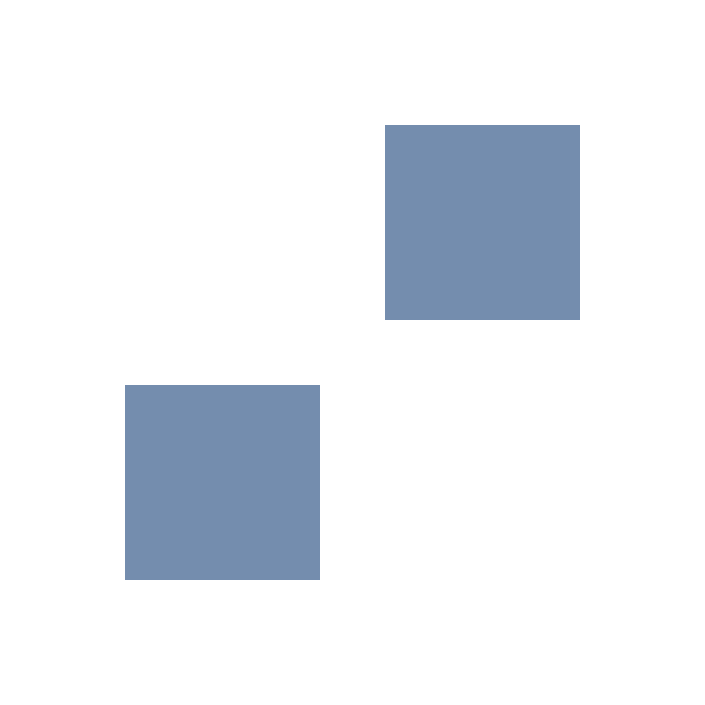}
  \includegraphics[width=.08\linewidth]{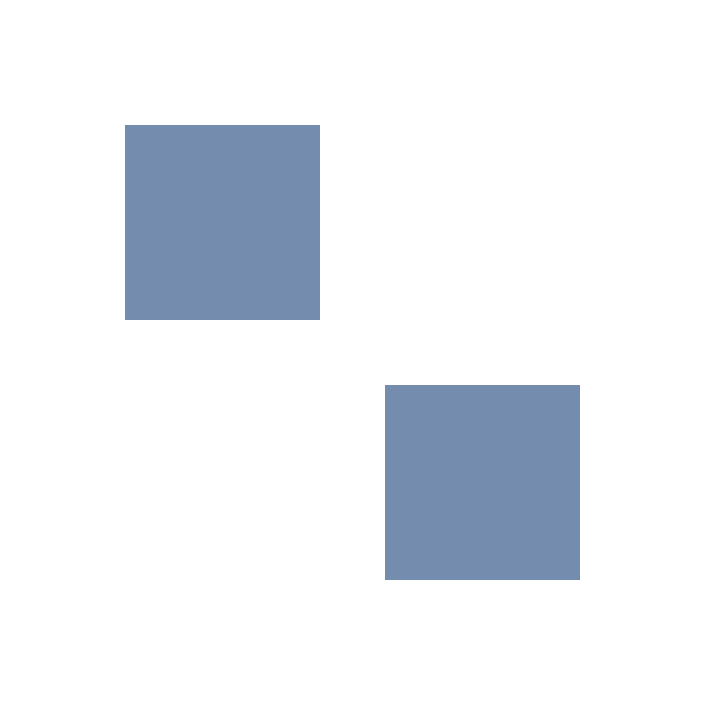}
  \caption{All four combinatorial arrangements of two axis-aligned disjoint rectangles in the plane. Two configurations are equivalent if, for every pair of rectangles, the same relative position holds in both configurations (namely, whether one lies strictly to the left of the other or strictly above the other).}
  \label{fig:n2}
\end{figure}

\begin{figure}[tbp]
  \raggedright

  \includegraphics[width=.08\linewidth]{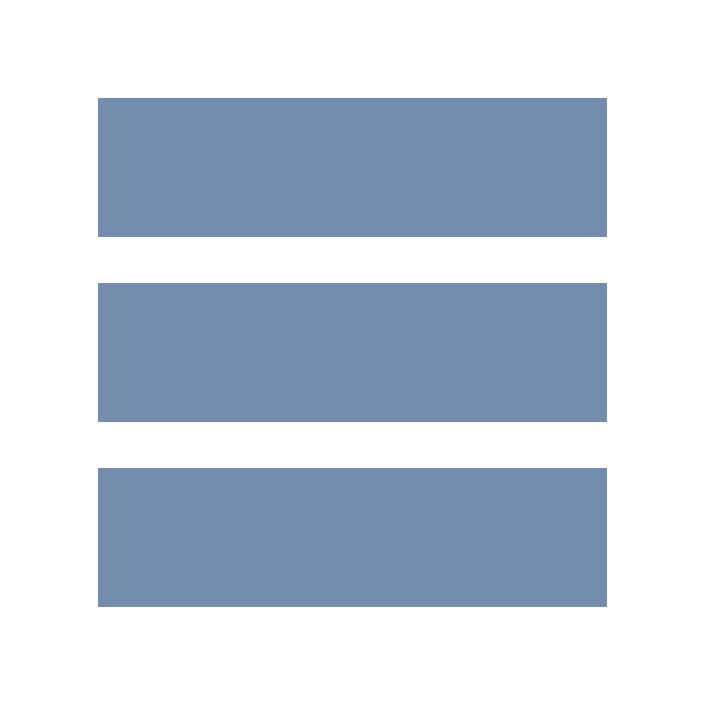} \includegraphics[width=.08\linewidth]{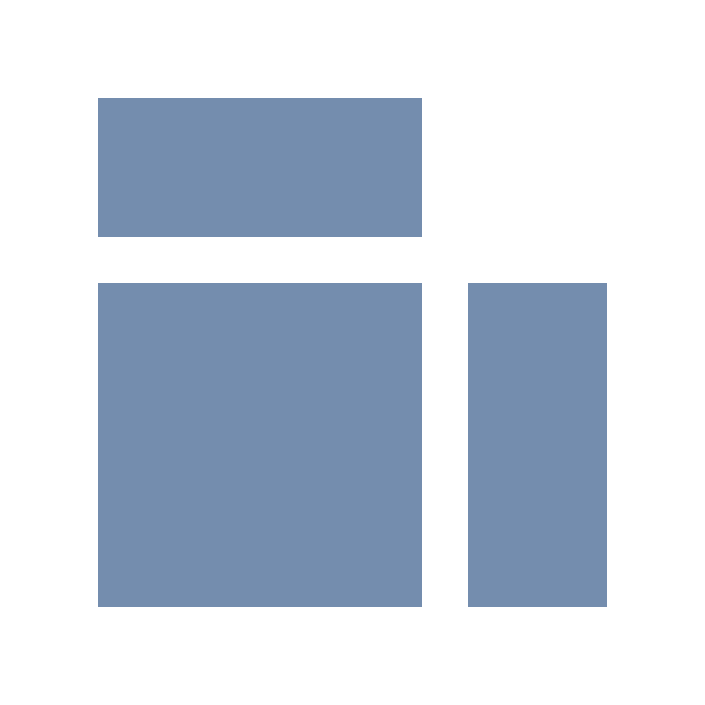} \includegraphics[width=.08\linewidth]{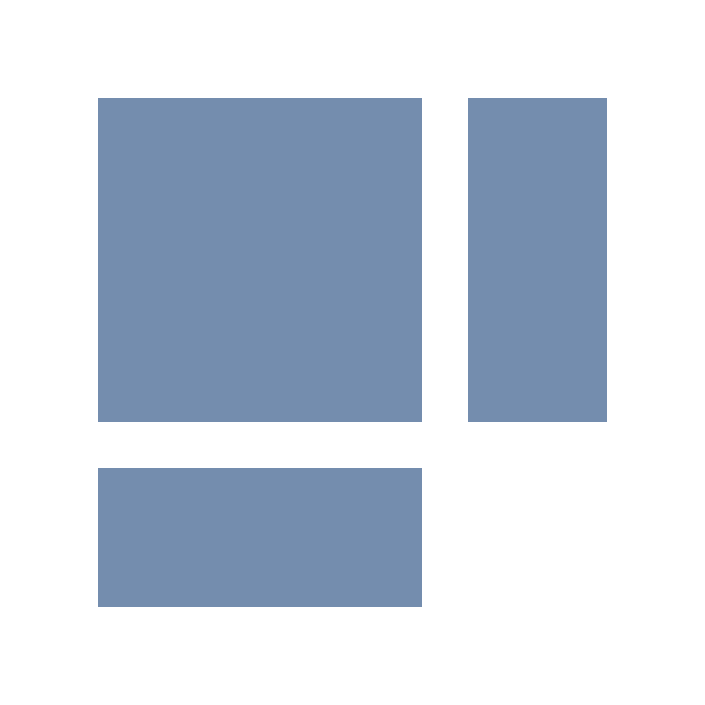} \includegraphics[width=.08\linewidth]{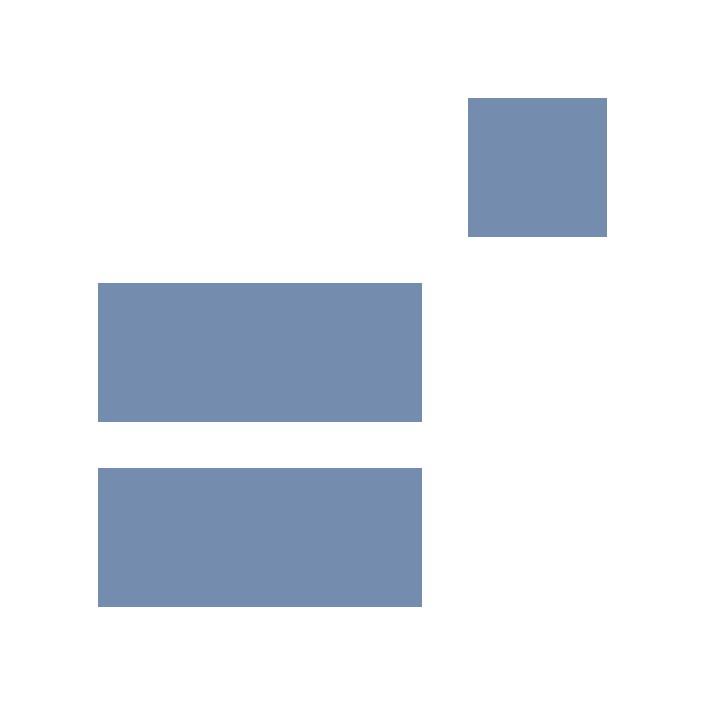} \includegraphics[width=.08\linewidth]{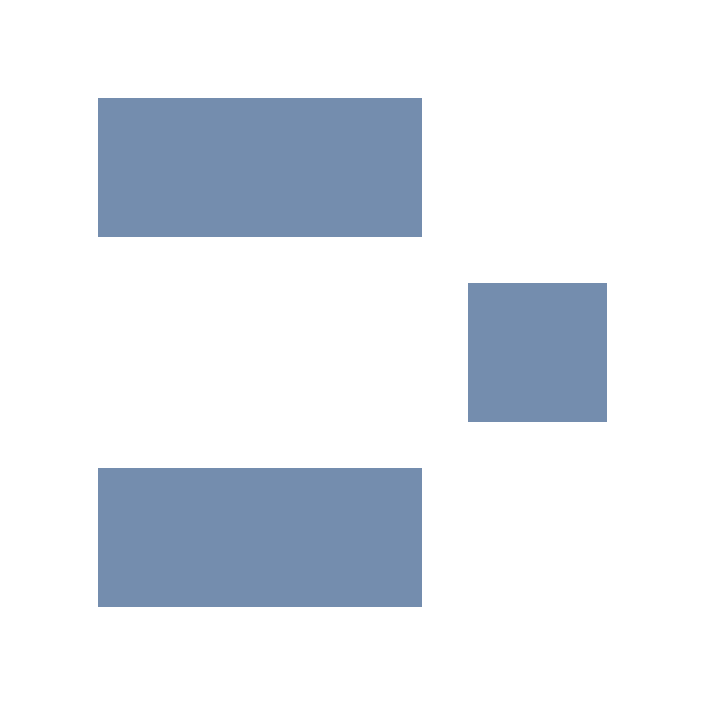} \includegraphics[width=.08\linewidth]{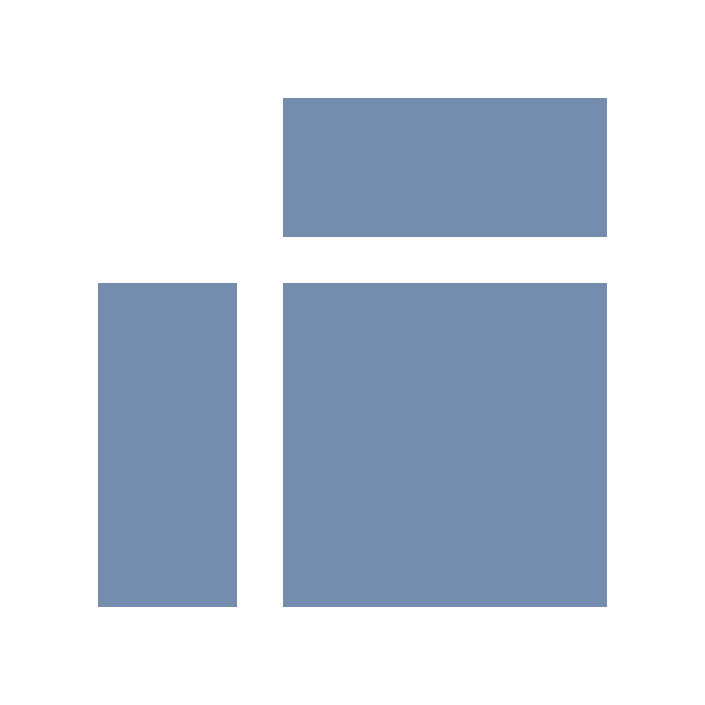} \includegraphics[width=.08\linewidth]{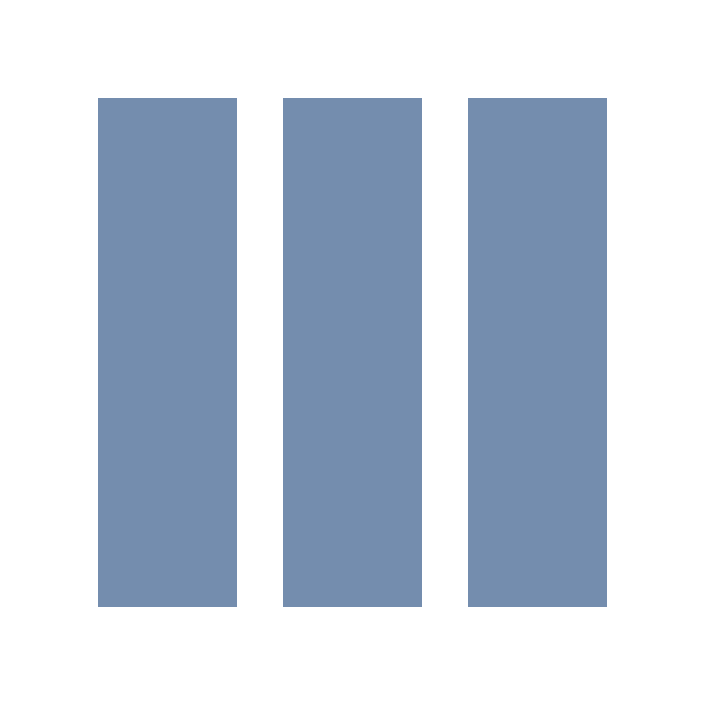} \includegraphics[width=.08\linewidth]{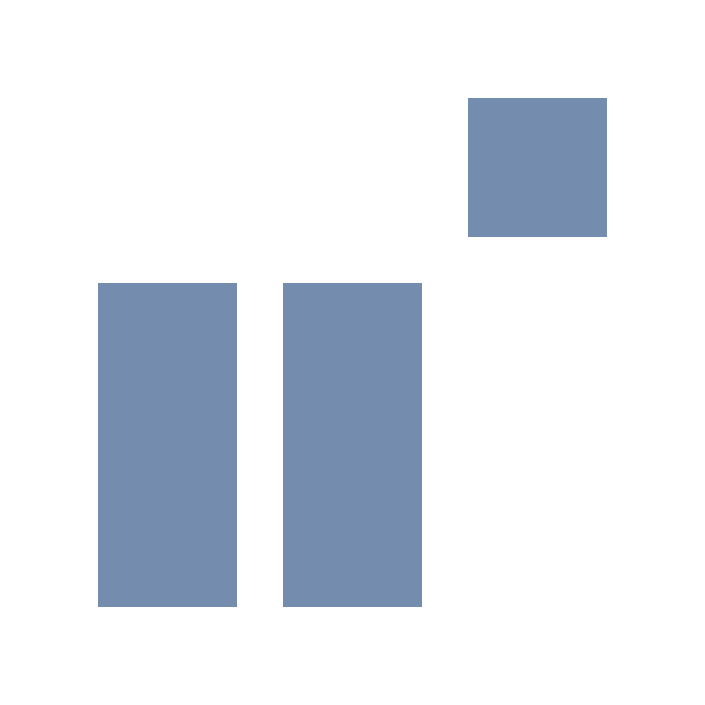} \includegraphics[width=.08\linewidth]{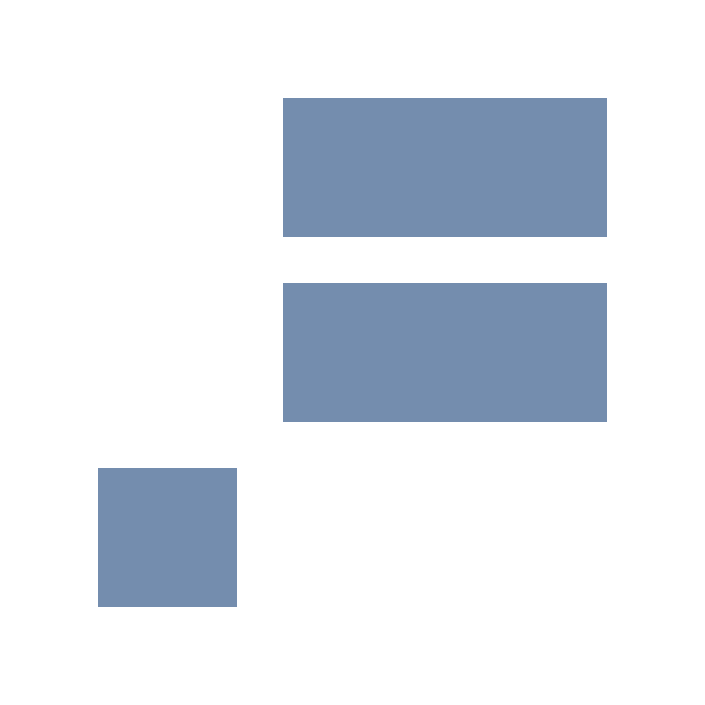} \includegraphics[width=.08\linewidth]{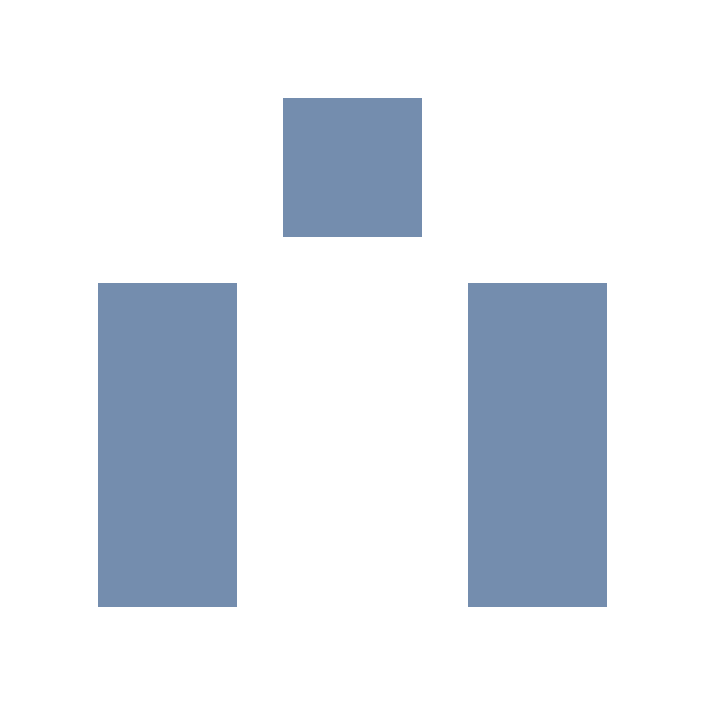} \includegraphics[width=.08\linewidth]{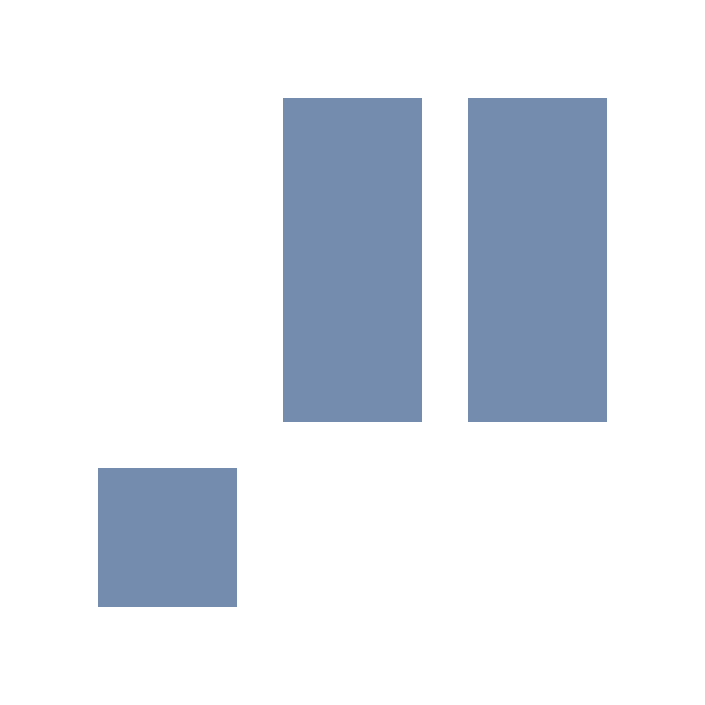} \includegraphics[width=.08\linewidth]{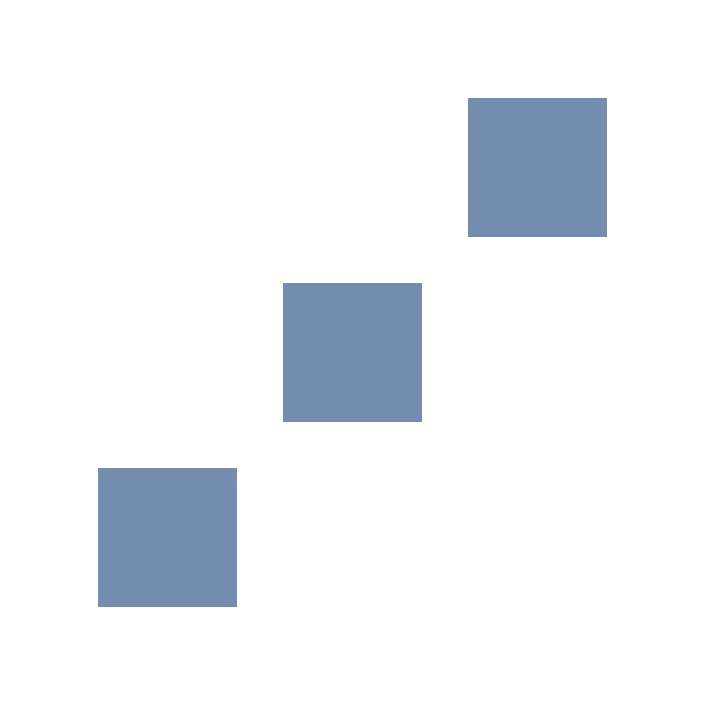} \includegraphics[width=.08\linewidth]{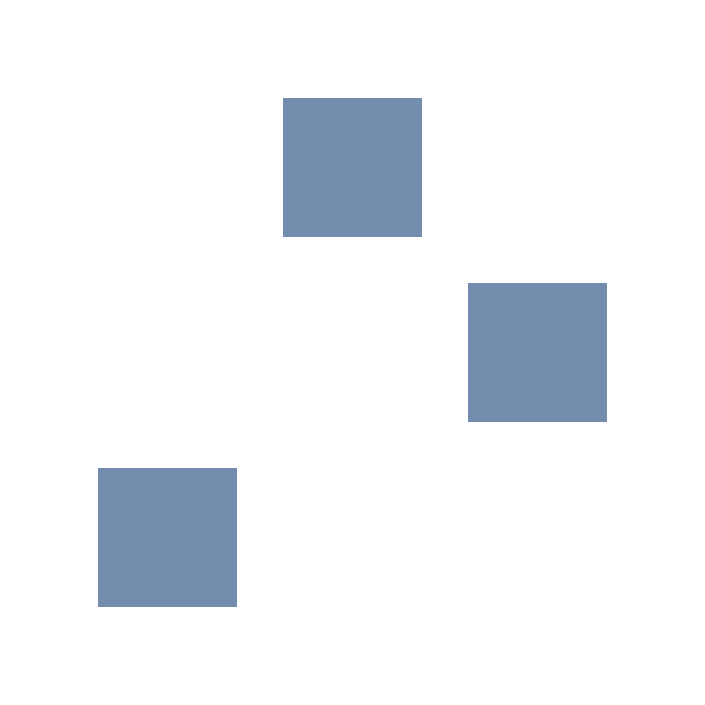} \includegraphics[width=.08\linewidth]{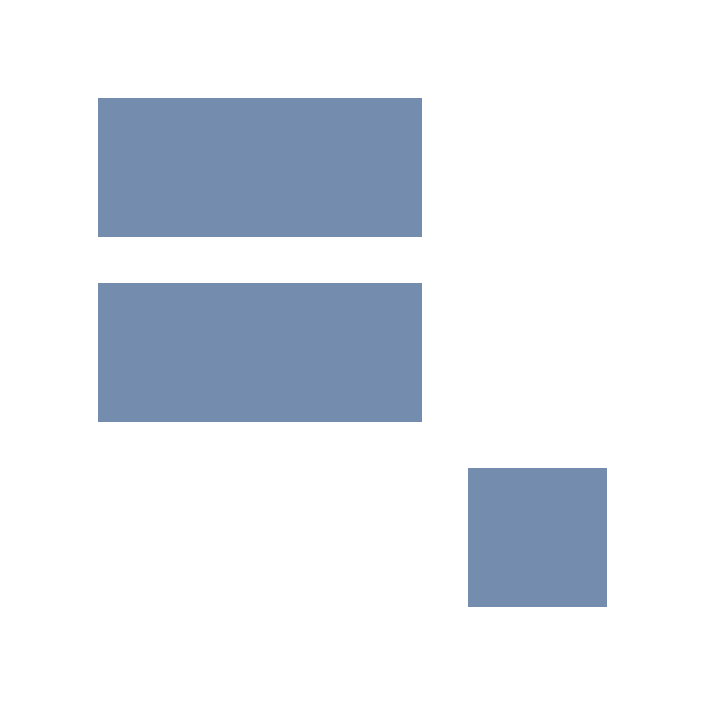} \includegraphics[width=.08\linewidth]{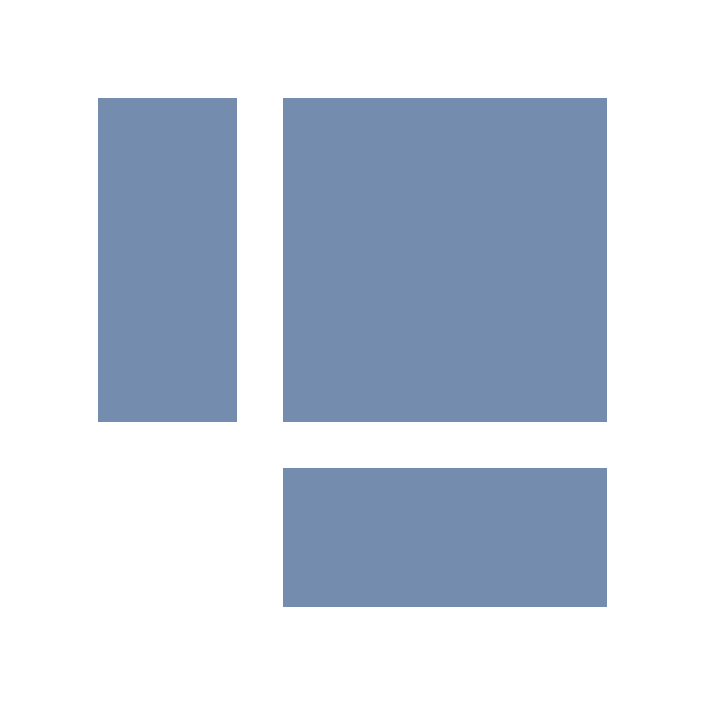} \includegraphics[width=.08\linewidth]{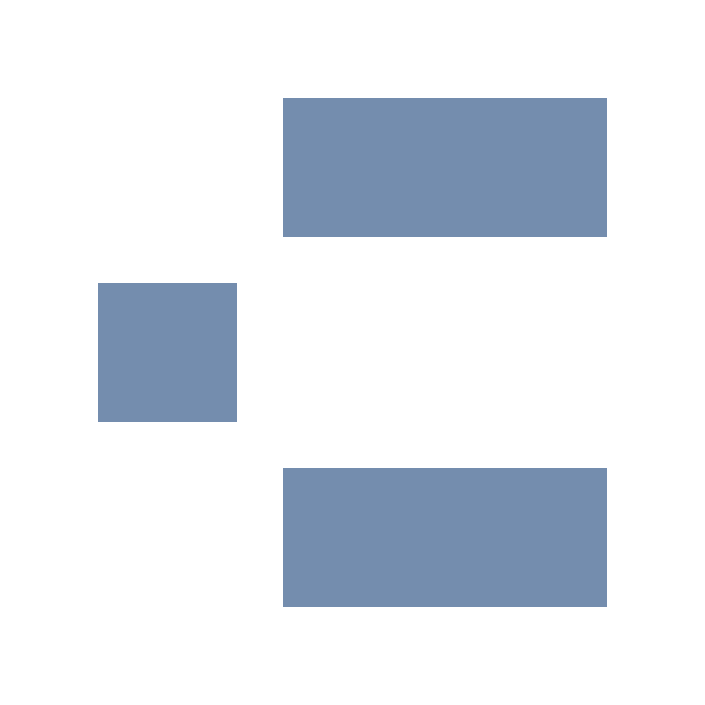} \includegraphics[width=.08\linewidth]{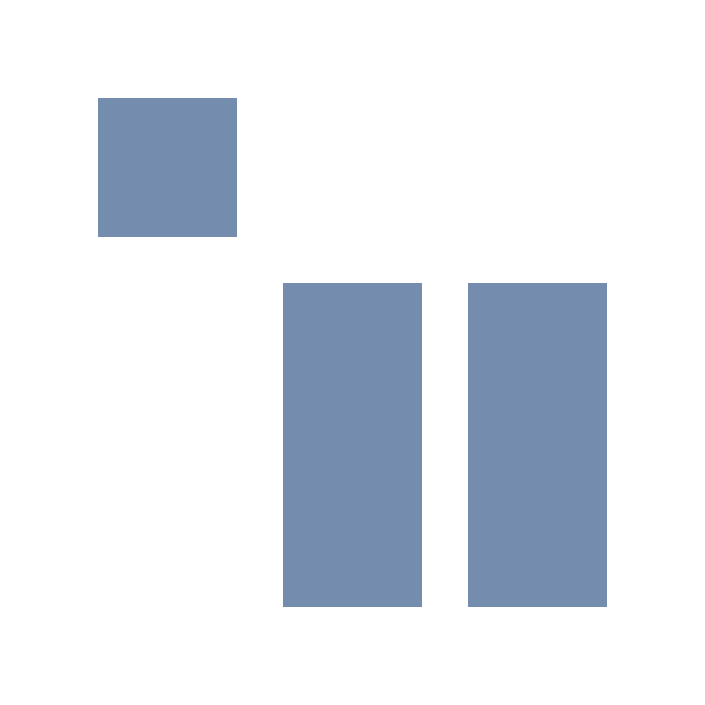} \includegraphics[width=.08\linewidth]{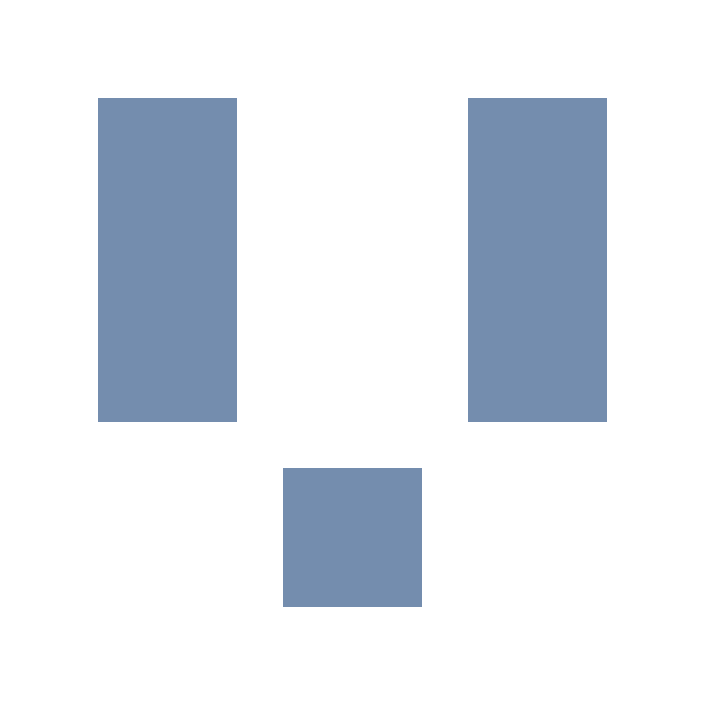} \includegraphics[width=.08\linewidth]{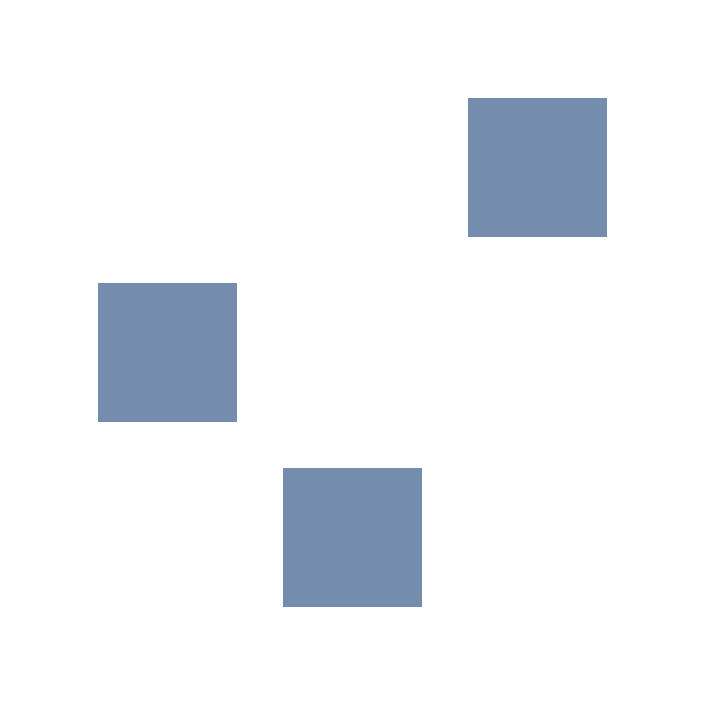} \includegraphics[width=.08\linewidth]{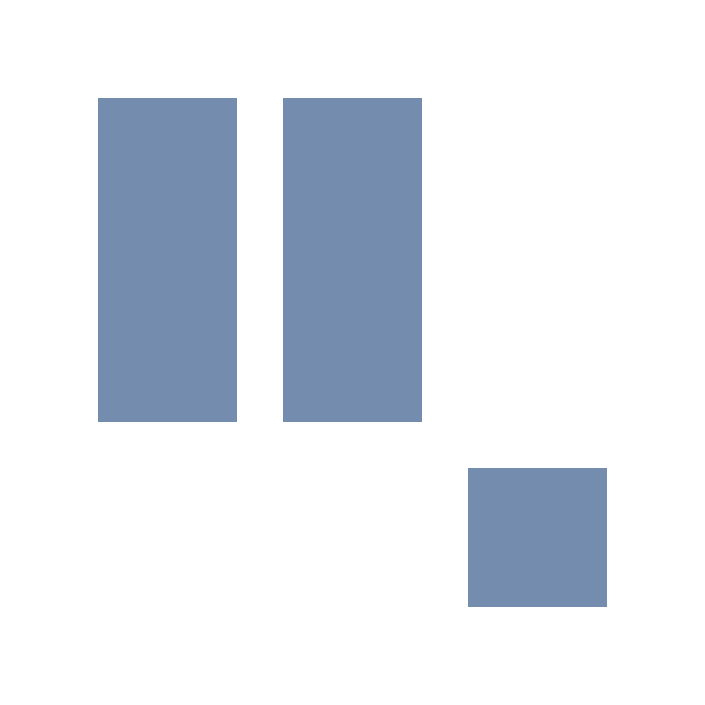} \includegraphics[width=.08\linewidth]{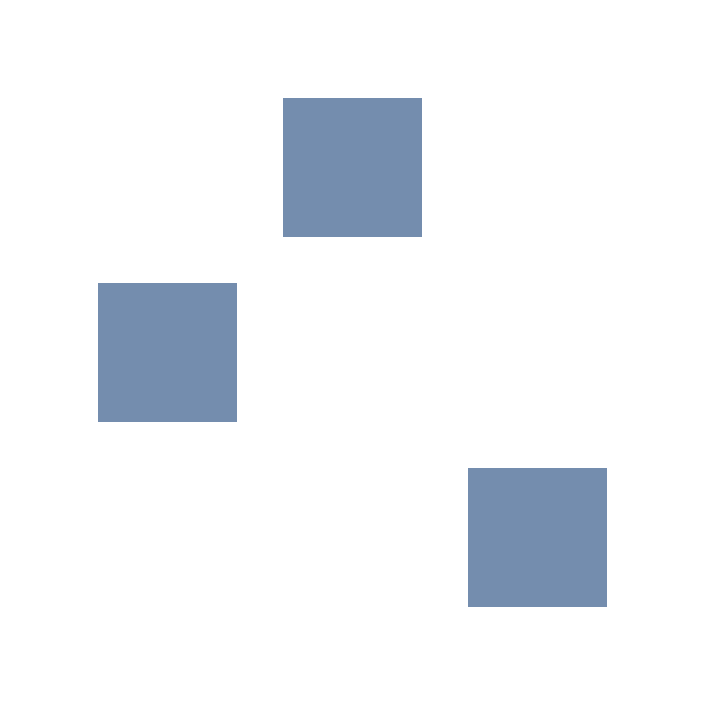} \includegraphics[width=.08\linewidth]{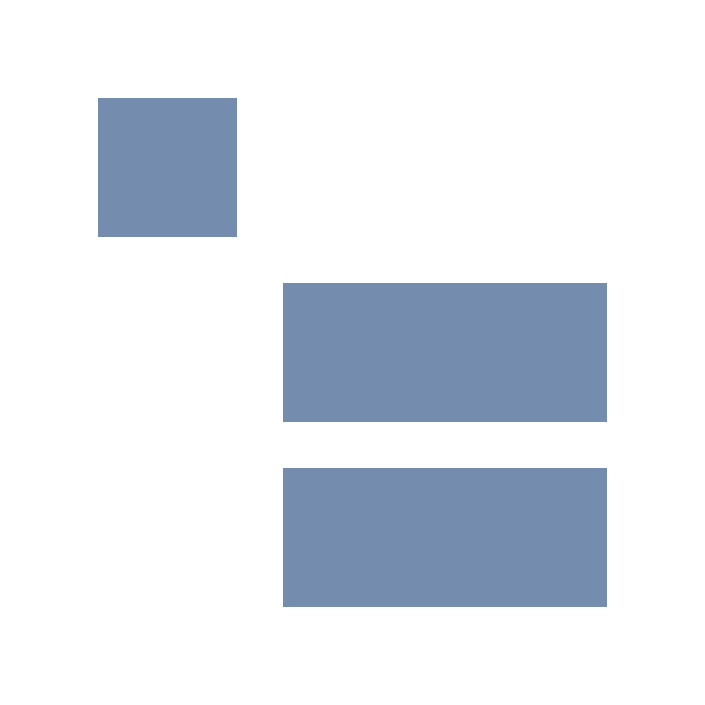} \includegraphics[width=.08\linewidth]{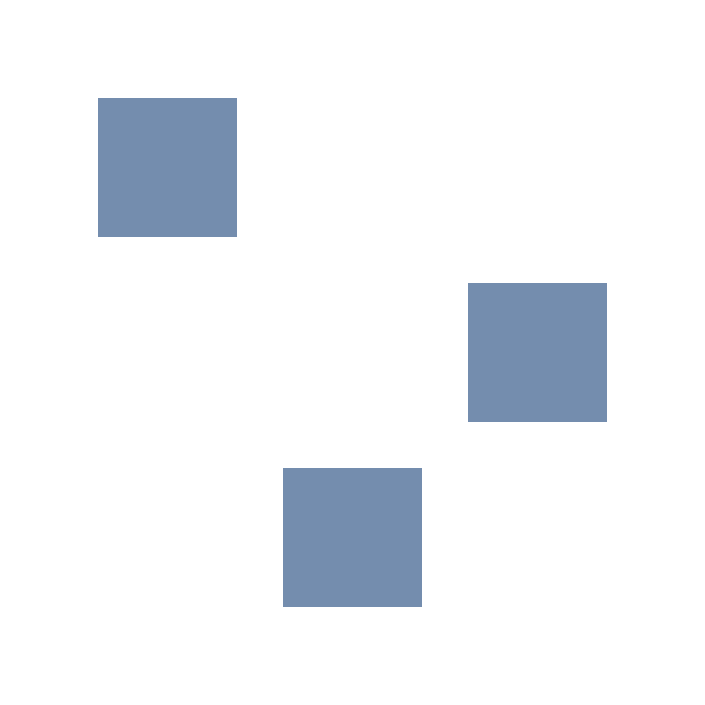} \includegraphics[width=.08\linewidth]{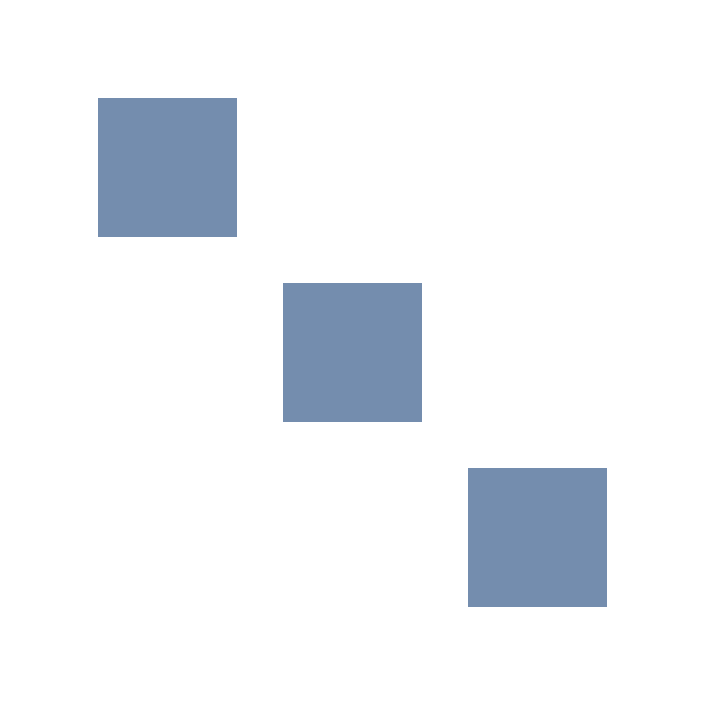} \includegraphics[width=.08\linewidth]{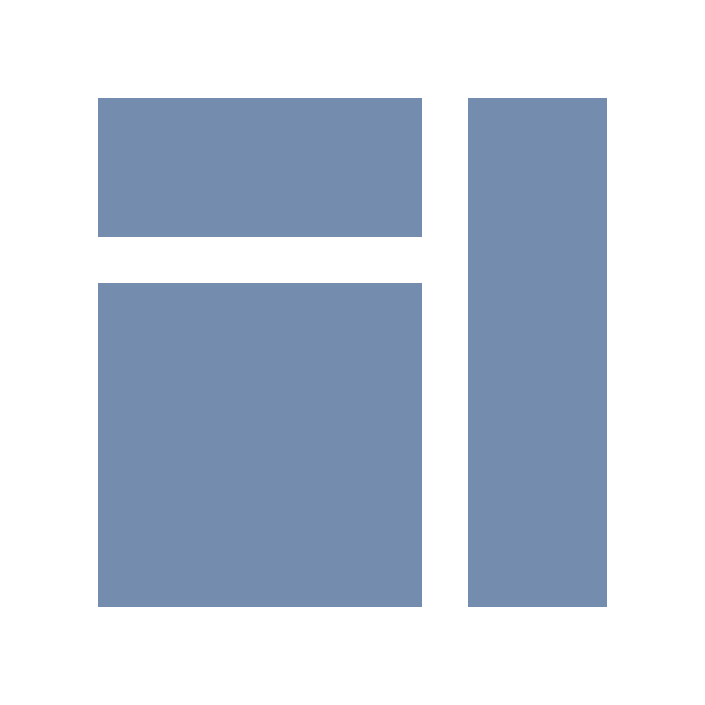} \includegraphics[width=.08\linewidth]{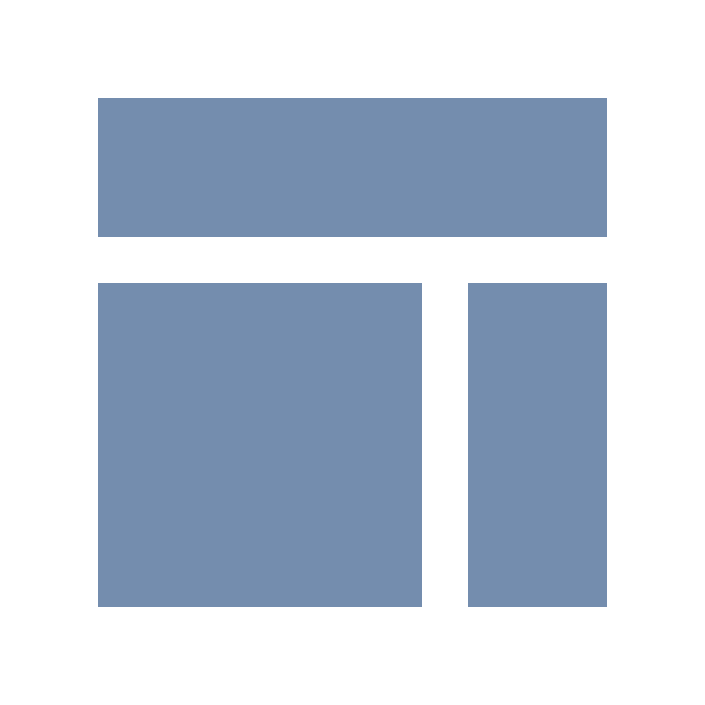} \includegraphics[width=.08\linewidth]{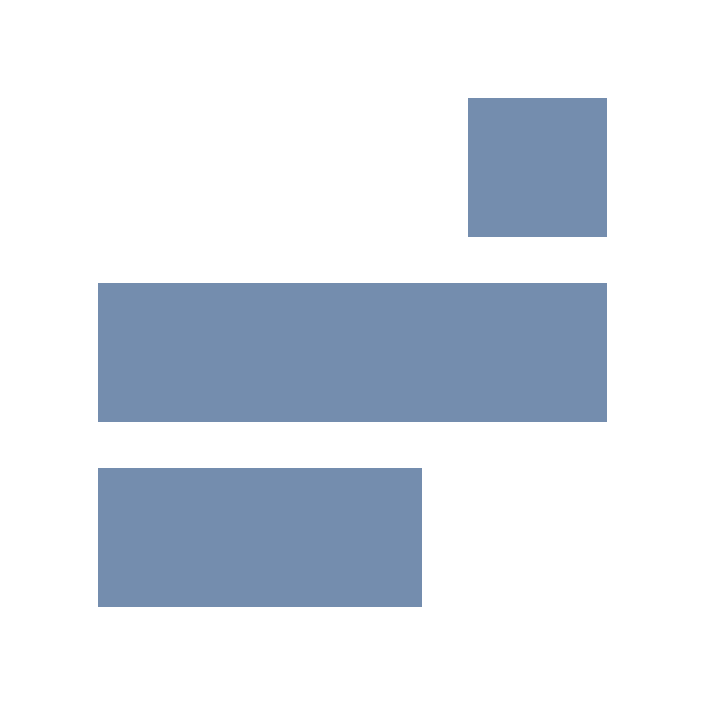} \includegraphics[width=.08\linewidth]{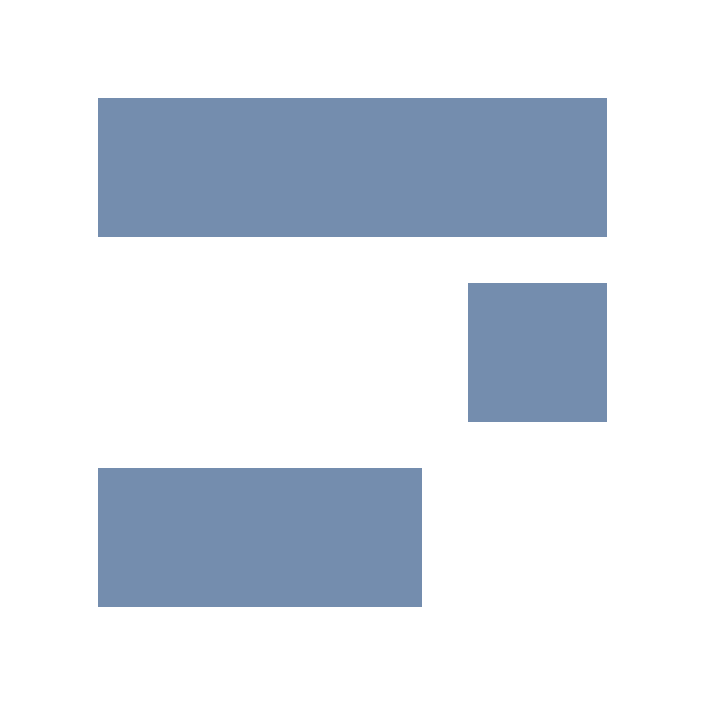} \includegraphics[width=.08\linewidth]{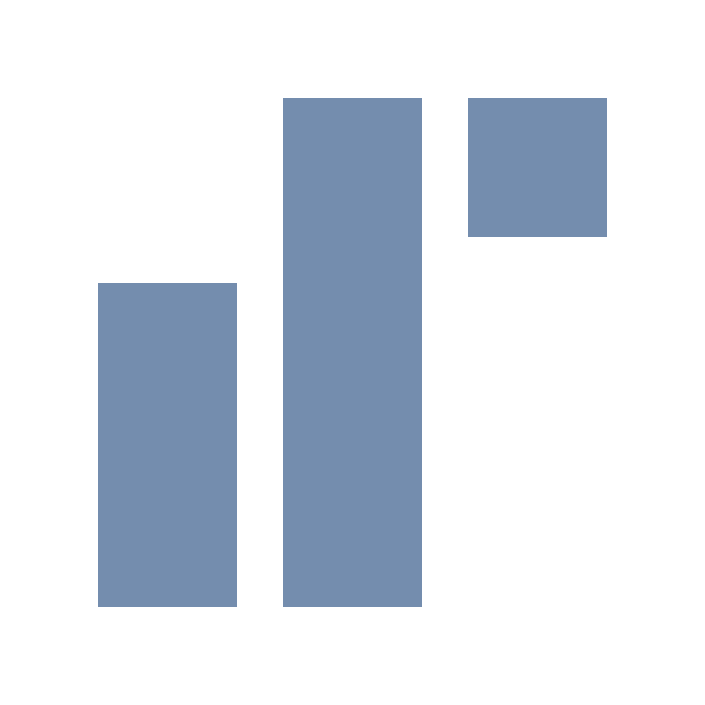} \includegraphics[width=.08\linewidth]{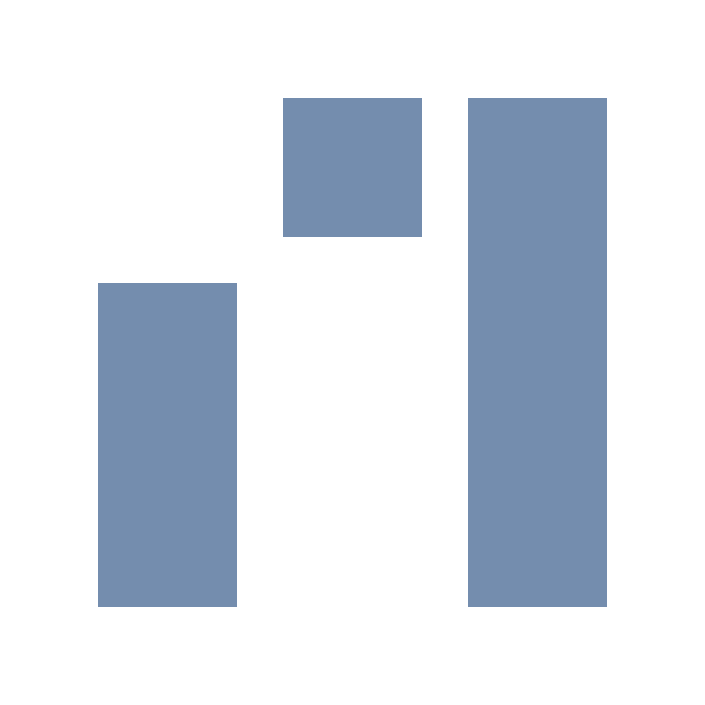} \includegraphics[width=.08\linewidth]{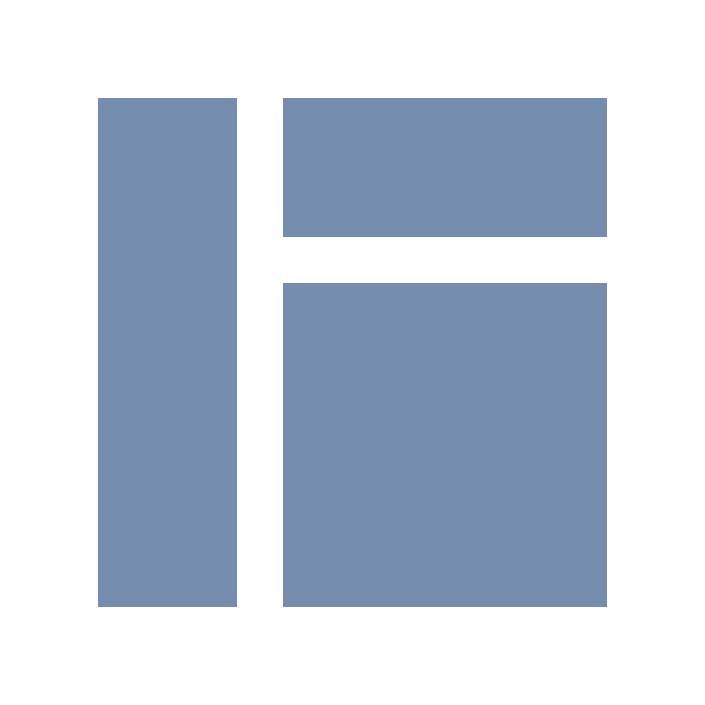} \includegraphics[width=.08\linewidth]{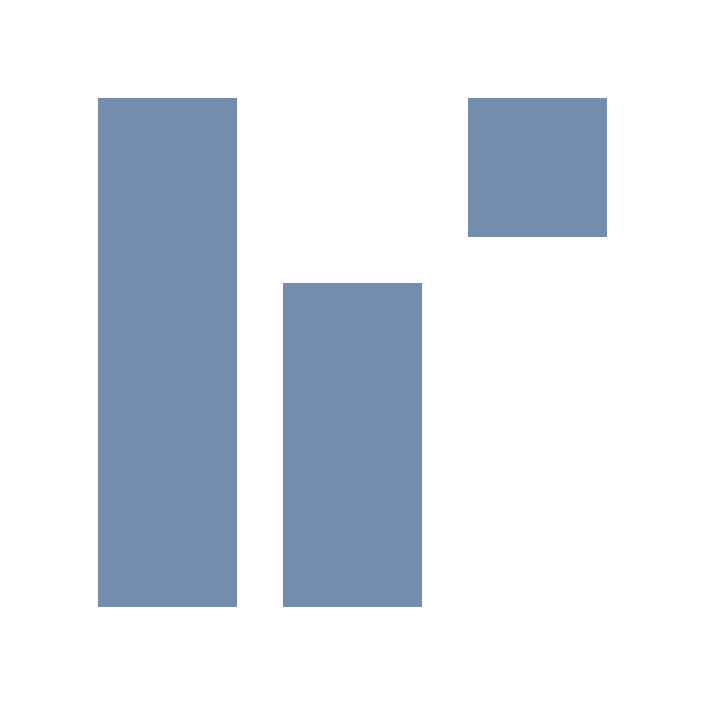} \includegraphics[width=.08\linewidth]{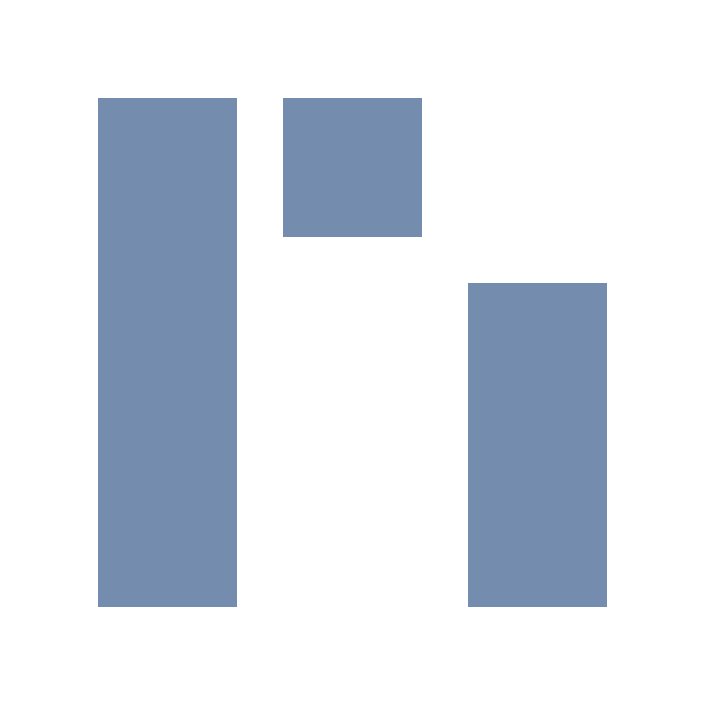} \includegraphics[width=.08\linewidth]{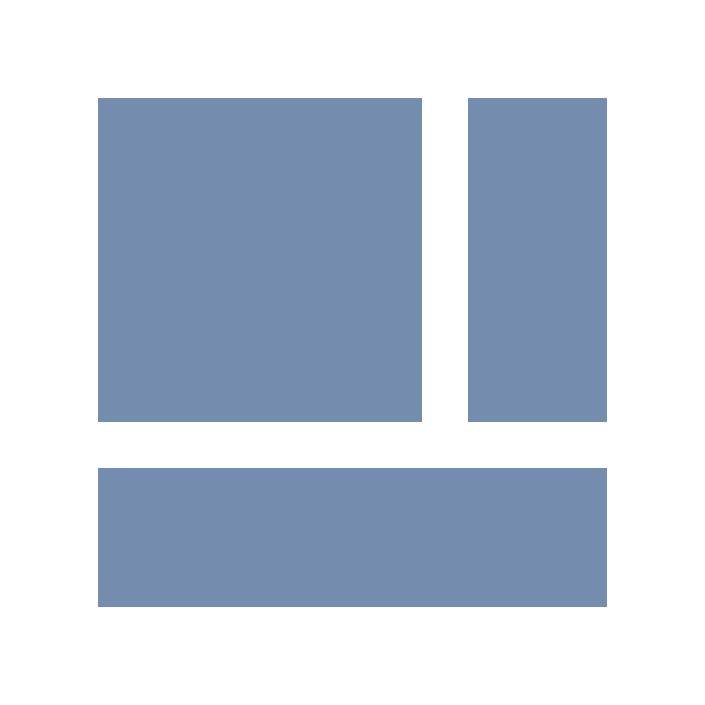} \includegraphics[width=.08\linewidth]{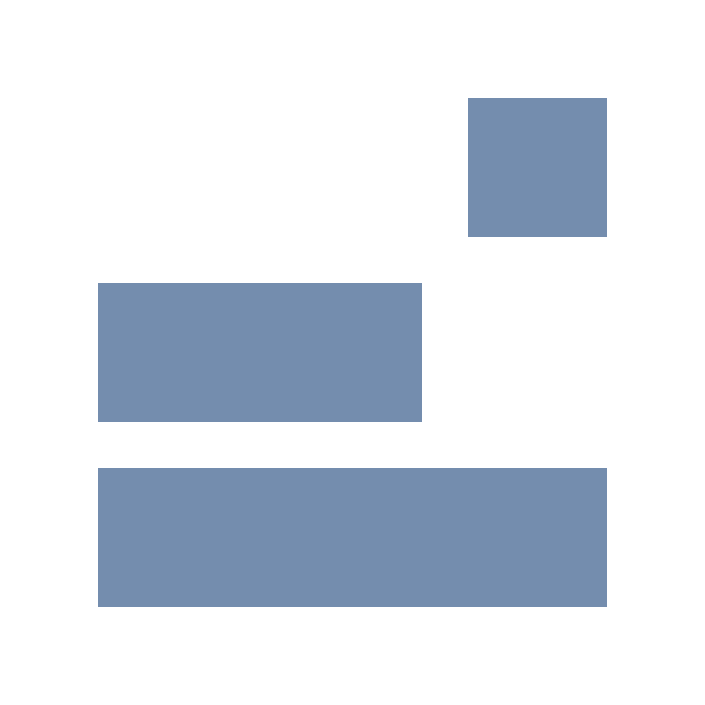} \includegraphics[width=.08\linewidth]{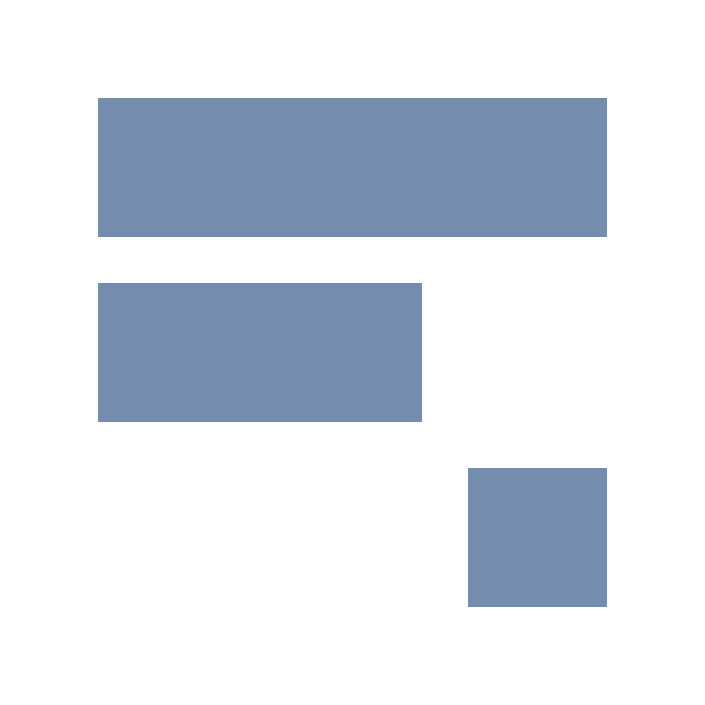} \includegraphics[width=.08\linewidth]{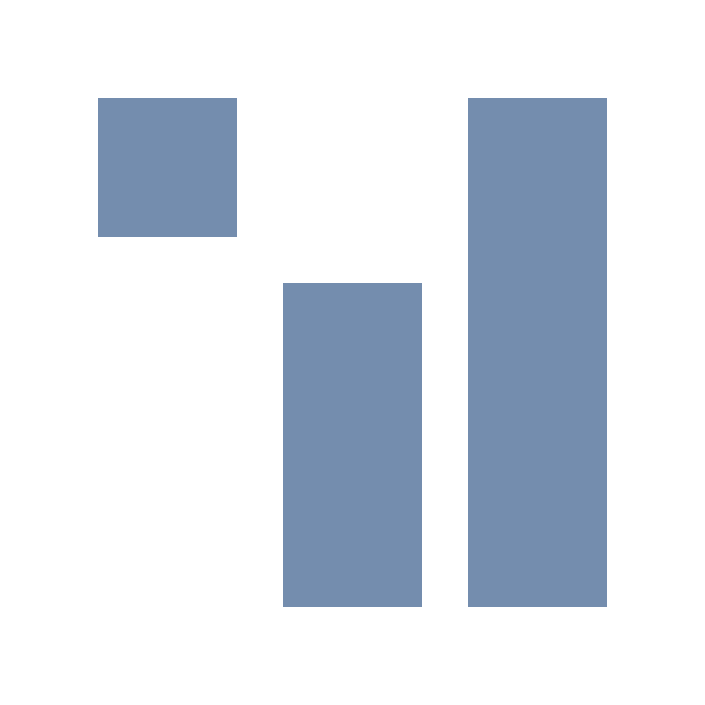} \includegraphics[width=.08\linewidth]{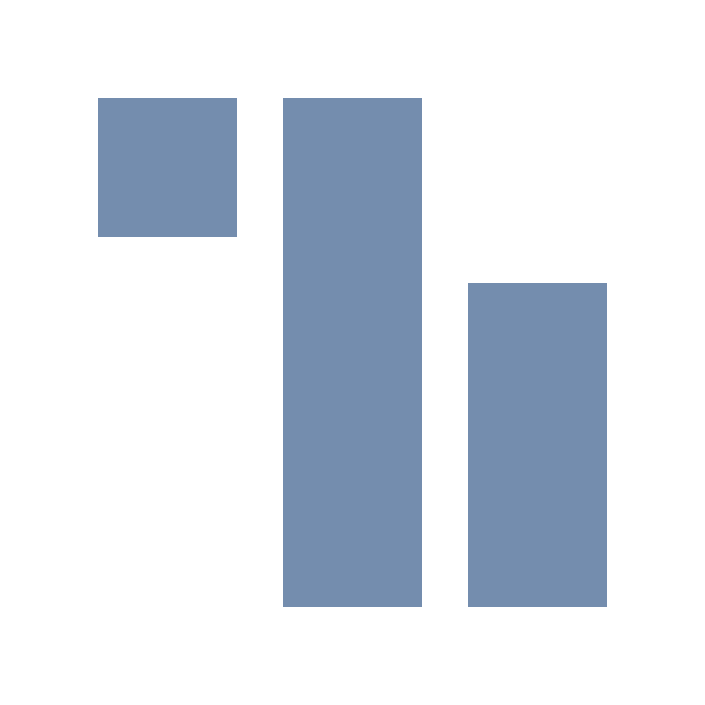} \includegraphics[width=.08\linewidth]{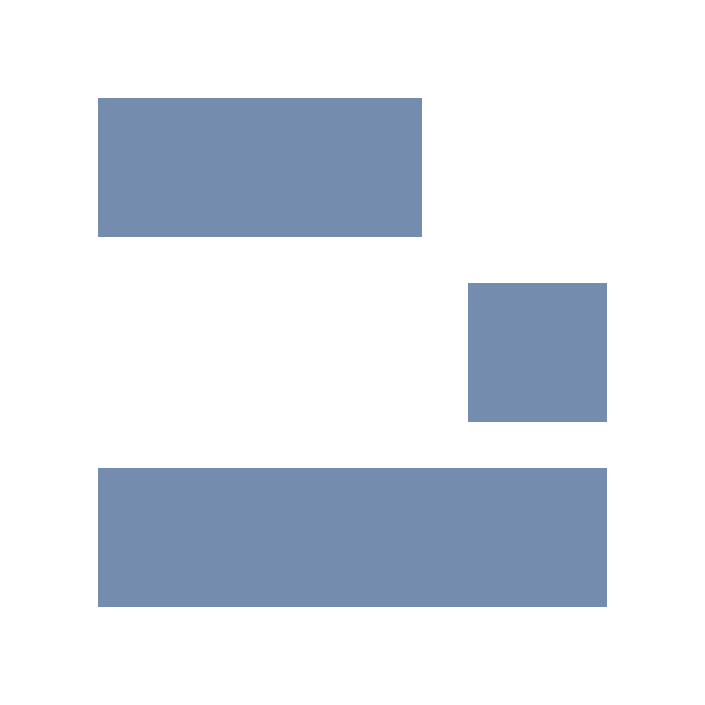} \includegraphics[width=.08\linewidth]{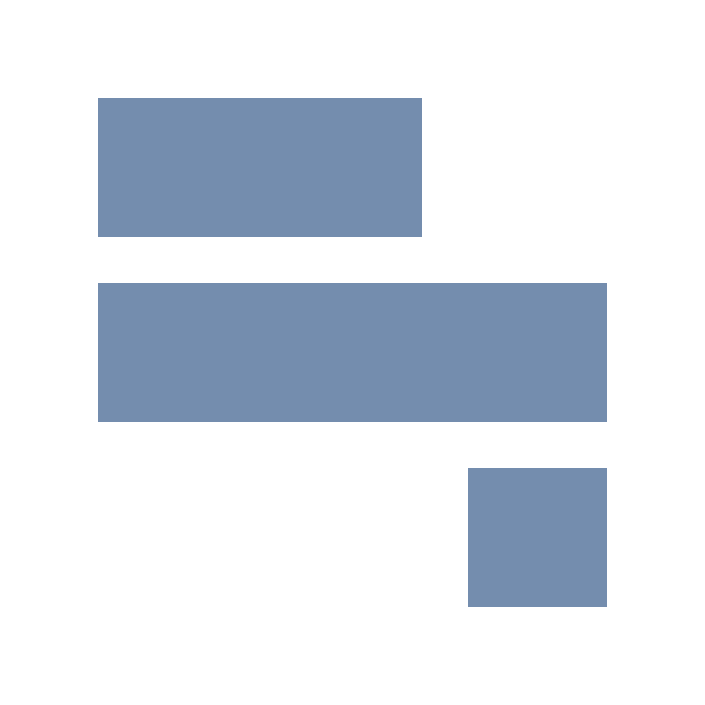}

  \caption{All 40 combinatorial arrangements of three axis-aligned disjoint rectangles in the plane.}
  \label{fig:n3}
\end{figure}

\begin{figure}
  %\vspace*{-3cm}
  \makebox[\textwidth][c]{
  \includegraphics[width=18cm]{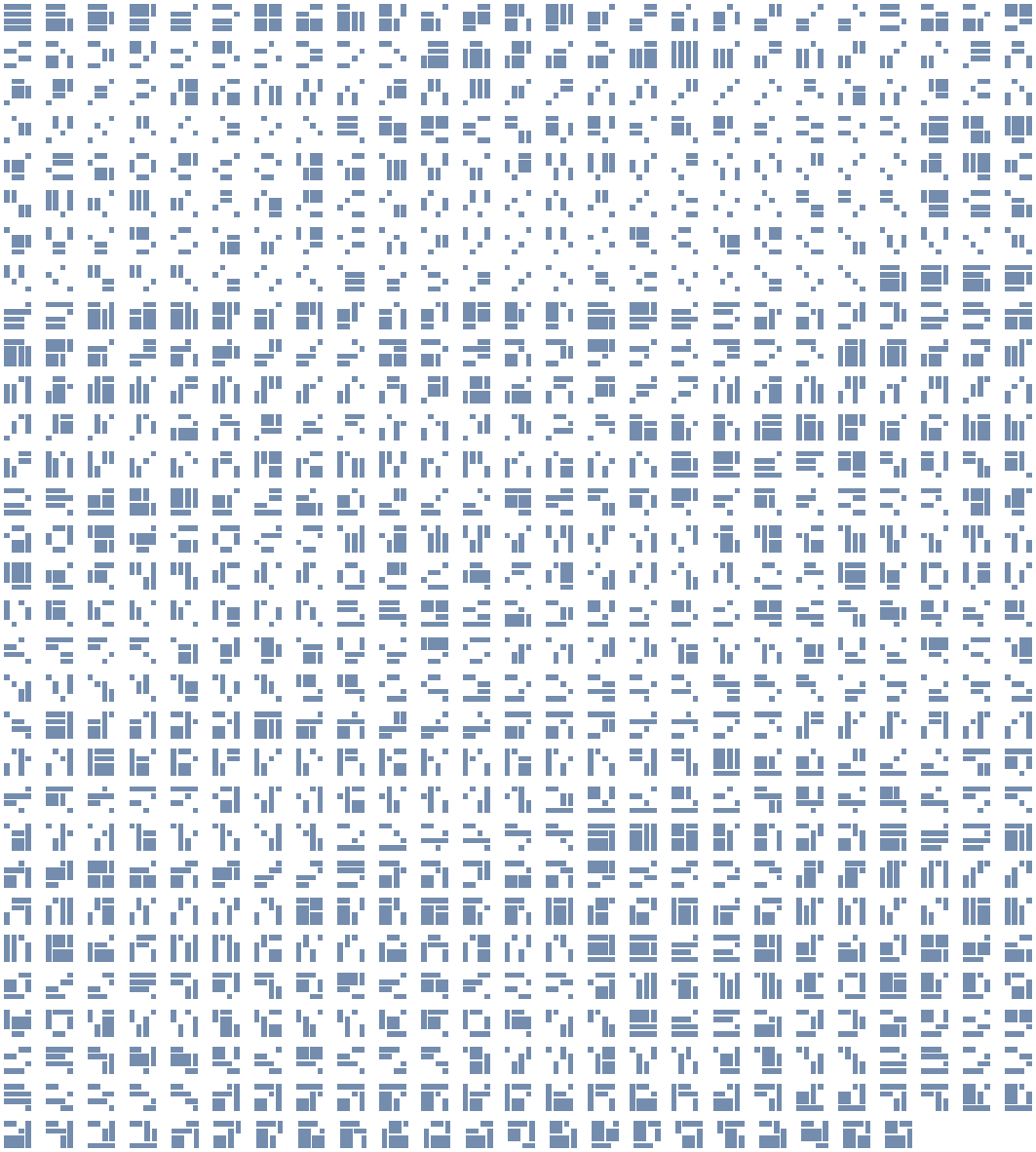}
  }

  \caption{All 772 combinatorial configurations of disjoint 4 rectangles in the plane.}
  \label{fig:n4}
\end{figure}

The combinatorics of points in an exact setting is simpler than its approximate counterpart.
\textcite{FoxNeuwirth_1962} described a CW-complex structure on $OC_n$, whose cells give a combinatorial description of configurations.
For example, with three points, the ordered configuration space~$OC_3$ admits
the partition formed by the following cells, up to permutation of~$z_1$, $z_2$ and~$z_3$:
\begin{align*}
   & \left\{ \re(z_1) < \re(z_2) < \re(z_3) \right\},                                             \\
   & \left\{ \re(z_1) < \re(z_2) = \re(z_3) \text{ and } \im(z_2) < \im(z_3) \right\},            \\
   & \left\{ \re(z_1) = \re(z_2) < \re(z_3) \text{ and } \im(z_1) < \im(z_2) \right\},            \\
   & \left\{ \re(z_1) = \re(z_2) = \re(z_3) \text{ and } \im(z_1) < \im(z_2) < \im(z_3) \right\}.
\end{align*}
Yet, determining a Fox--Neuwirth cell containing a triple~$(z_1,z_2,z_3)$ from a numerical approximation, even a certified one, may be impossible because equality cannot be decided. This motivates using arrangement cells instead,
which are open in~$OC_n$ but lead to much richer combinatorics.

The combinatorial diversity of arrangement cells can be pictured by those coming from arrangements of disjoint axis-aligned rectangles (or \emph{boxes}) in the plane, as they arise in interval arithmetic.
In this context, a configuration $(z_1,\dotsc,z_n) \in OC_n$ is described by an interval for each coordinate $\re(z_1)$, $\im(z_1)$, $\re(z_2)$, etc., so each~$z_i$ lies in a box, and we assume that the boxes are pairwise disjoint.
Two boxes can be separated by a vertical line, or a horizontal line.
So a tuple of disjoint boxes induces an arrangement.
Figures~\ref{fig:n2}, \ref{fig:n3}, and~\ref{fig:n4} display all the arrangements that can be obtained in this way, with~2,~3 and~4 points, respectively.

\section{Computing braids from arrangements}
\label{sec:braids}

A sequence of arrangement cells covering a path~$F:[0,1]\to OC_n$ entirely determines the homotopy type of~$F$ (Lemma~\ref{lem:homotopy-arrangement}),
and Algorithm~\ref{algo:cover} computes such a sequence.
This section focuses on the next step: computing the braid induced by~$F$.

\subsection{Permutation points and elementary braids}
\label{sec:perm-points-elem}

To compute a braid from a sequence of arrangements covering a path, we use a piecewise linear path
whose vertices are permutation points.
Given two permutations~$\pi$ and~$\phi$ of~$[n]$, the \emph{permutation point~$p(\pi,\phi)$}
is the point in~$OC_n$ defined by
\[ p(\pi,\phi) = \left( \pi(1) + \phi(1) \I, \dotsc,  \pi(n) + \phi(n) \I   \right). \]
Note that~$\sigma \cdot p( \pi, \phi) = p(\pi \sigma^{-1}, \phi \sigma^{-1})$ for any permutation~$\sigma \in S_n$.

\begin{lemma}
  \label{lem:arrangement-permutation-point}
  Any arrangement cell contains a permutation point.
\end{lemma}

\begin{proof}
  Same argument as the nonemptiness proof in Lemma~\ref{lem:convex-open-cover}.
\end{proof}

We also consider the point~$q(\pi) = \left( \pi(1),\dotsc,\pi(n) \right) \in OC_n$.
The collection of these points forms the fiber of the projection $OC_n \to C_n$ above the base point~$[n] \in C_n$.
A permutation point~$p(\pi,\phi)$ is connected to~$q(\pi)$ by a straight-line path in~$OC_n$.
%\[ t\mapsto \left( \pi(1) + t \phi(1) \I, \dotsc,  \pi(n) + t \phi(n) \I \right). \]
Given a path~$P$ between two permutation points~$p(\pi, \phi)$ and~$p(\pi', \phi')$,
we can extend~$P$ as a path~$P'$ from~$q(\pi)$ to~$q(\pi')$:
\[ q(\pi) \to p(\pi, \phi) \overset{P}{\to} p(\pi', \phi') \to q(\pi'). \]
This denotes the path that starts from~$q(\pi)$, goes to~$p(\pi, \phi)$ in a straight line,
follows~$P$ to reach~$p(\pi', \phi')$, and then goes to~$q(\pi')$ in a straight line.
The projection of~$P'$ on~$C_n$ induces a loop from and back to the base point~$[n]$, therefore its homotopy class is described by a combinatorial~$b \in B_n$.
Conversely, a braid~$b \in B_n$ is the homotopy class of a loop in~$C_n$.
We may lift this loop as a path~$P'$ in~$OC_n$ starting from~$q(\pi)$
and ending in another point of the fiber over~$[n]$, some~$q(\pi')$.
Given two permutations~$\phi$ and~$\phi'$, we may form the extended path
\[ p(\pi, \phi) \to q(\pi) \overset{P'}{\to} q(\pi') \to p(\pi', \phi'), \]
which is well defined up to homotopy.
We denote this path~$p(\pi, \phi) \overset{b}{\to} p(\pi', \phi')$.
These two constructions, from the homotopy class of~$P$ to~$b$, and from~$b$, $\pi$, $\phi$ and~$\phi'$ to the homotopy class of~$P$, are inverse to one another.
Moreover, they are compatible with composition:
\[ p(\pi, \phi) \overset{b}{\to}  p(\pi', \phi') \overset{b'}{\to} p(\pi'', \phi'') \quad \homeq \quad  p(\pi, \phi) \overset{b\cdot b'}{\to} p(\pi'', \phi''). \]

Given two permutations~$\pi$, $\phi$, the piecewise-linear path
$q(\pi) \to p(\pi, \phi) \to p(\pi, \operatorname{id}) \to q(\pi)$
lies in~$OC_n$ and induces the trivial braid: since the real parts remain constant, the path is contractible.
However, the linear path
$p(\pi, \phi) \to p(\operatorname{id}, \phi)$, which also lies in~$OC_n$, may induce a nontrivial braid.
Let~$B(\pi, \phi)$ denote its associated braid, called an \emph{elementary braid}. In other words,
\[ p(\pi, \phi) \overset{\text{straight line}}\to p(\identity, \phi)  \quad\equiv_\text{homotopy}\quad p(\pi, \phi) \overset{B(\pi, \phi)}\to p(\identity, \phi). \]

% Given permutations $\pi$, $\phi$, $\pi'$ and~$\phi'$,
% the linear path
% $p(\pi,\phi) \to p(\pi, \phi')$ is trivial, for the same reason as above,
% and the linear path
% $p(\pi, \phi) \to p(\pi', \phi)$ may induce a nontrivial braid.
% By invariance under the action of the symmetric group,
% this braid is~$B(\pi \pi'^{-1}, \phi \pi^{-1})$.

\begin{lemma}\label{lem:basic-paths}
  Let~$A$ be an arrangement and let~$\pi$, $\phi$, $\pi'$ and~$\phi'$ be permutations of~$[n]$ such that the permutation points~$p(\pi, \phi)$ and~$p(\pi', \phi')$ lie in~$A$.
  Then the straight-line path~$p(\pi, \phi) \to p(\pi', \phi')$ is contained in~$OC_n$ and its associated braid is~$B(\pi \pi'^{-1}, \phi \pi'^{-1})$.
\end{lemma}

\begin{proof}
  By definition of an arrangement cell, if both~$p(\pi, \phi)$ and~$p(\pi', \phi')$ are in~$A$, then so is~$p(\pi', \phi)$.
  So the path $p(\pi, \phi) \to p(\pi', \phi')$ is homotopic to
  \[ p(\pi, \phi) \to p(\pi', \phi) \to p(\pi', \phi'), \]
  by convexity of~$A$.
  Under the action of~$\pi'^{-1}$, which does not change the associated braid, this path becomes
  \[ p(\pi \pi'^{-1}, \phi \pi'^{-1}) \to p(\identity, \phi \pi'^{-1}) \to p(\identity, \phi' \pi'^{-1}). \]
  The first part of the path induces the braid~$B(\pi \pi'^{-1}, \phi \pi'^{-1})$, by definition, while the second part induces the trivial braid, since the real parts remain fixed.
\end{proof}

The following statement gives identities to compute~$B(\pi, \phi)$ in terms of a decomposition of~$\pi$ as a product of elementary transpositions~$(i \ i+1)$.

\begin{lemma} \label{lem:b-braid}
  For any permutations~$\pi$, $\pi'$ and $\phi$:
  \begin{enumerate}[(i)]
    \item $B(\identity, \phi) = 1$;
    \item $B(\pi'\pi, \phi) = B(\pi',\phi \pi^{-1})B(\pi,\phi)$;
    \item $B( (i\ i+1), \phi ) = \sigma_i$ if $\phi(i) > \phi(i+1)$, for any~$1\leq i < n$;
    \item $B( (i\ i+1), \phi ) = \sigma_i^{-1}$ if $\phi(i) < \phi(i+1)$, for any~$1\leq i < n$.
  \end{enumerate}
\end{lemma}

\begin{proof}
  The first point is trivial since the linear path~$p(\identity, \phi)\to p(\identity,\phi)$ is constant.

  For the second point,
  consider the piecewise straight-line path~$P$
  \[  p(\pi'\pi, \phi) \to p(\pi,\phi) \to p(\identity, \phi), \]
  and the arrangement~$A = (\varnothing, \prec_\phi)$ where~$\varnothing$ is the empty partial order and
  where~$\prec_\phi$ is the total order defined by~$i \prec_\phi j$ when~$\phi(i) < \phi(j)$.
  The path~$P$ entirely lies in the cell attached to this arrangement.
  The convexity of this cell (Lemma~\ref{lem:convex-open-cover}) implies that $P$ is homotopic to the straight-line path $p(\pi'\pi, \phi) \to p(\identity, \phi)$,
  which induces the braid~$B(\pi'\pi,\phi)$, by definition.

  On the other hand, the part $p(\pi,\phi) \to p(\identity, \phi)$
  induces the braid~$B(\pi,\phi)$, by definition,
  and the part $p(\pi'\pi, \phi) \to p(\pi,\phi)$
  induces the same braid as the path
  \[  p(\pi', \phi \pi^{-1}) \to p(\identity,\phi \pi^{-1}), \]
  by invariance under the action of the permutation group.
  This latter path induces the braid~$B(\pi', \phi \pi^{-1})$, by definition.
  This gives the identity $B(\pi'\pi, \phi) = B(\pi',\phi \pi^{-1})B(\pi,\phi)$.

  For the third point, let~$\tau$ denote the transposition~$\left( i\ i+1 \right)$, and consider the path~$P$
  \[ q( \tau ) \to p(\tau, \phi) \to p(\identity, \phi) \to q(\identity). \]
  It is covered by three arrangement cells: $(\prec_\tau, \varnothing)$, $A$, and~$(\prec_{\identity}, \varnothing)$, where~$A$
  is the arrangement~$(\prec_{\re}, \prec_{\im})$ with~$\prec_{\re} = (\prec_\tau \cap \prec_{\identity})$ and~$\prec_{\im}$ contains the only relation~$i + 1 \prec_{\im} i$.
  We check easily that these three arrangements also cover the path defining~$\sigma_i$ (Figure~\ref{fig:braid-generator}).
  The last point is similar.
\end{proof}

\subsection{Braids from sequences of arrangements}

\begin{figure}[tb]
  \def\motion#1#2{
    \tikzmath{integer \iter; \iter=#2/18;}
    \draw[sync] (#1 |- top) node[above] (time#2) {$t = \tfrac{\iter}{5}$} -- (#1 |- bot);
    \begin{scope}[shift=(#1 |- fig), label distance=-2, radius = 2]
      \path[arr] (2,0)  arc[start angle=0, end angle=#2]  node[label={#2:\small 1}] (A) {};
      \path[arr] (0,2)  arc[start angle=90, end angle=90 + #2]  node[inner sep=3pt,  label={90+#2:\small 2}] (B) {};
      \path[arr] (-2,0)  arc[start angle=180, end angle=180+#2]  node[label={180+#2:\small 3}] (C) {};
      \path[arr] (0,-2)  arc[start angle=270, end angle=270+#2]  node[label={270+#2:\small 4}] (D) {};

      \node[dot] at (A) {};
      \node[dot, inner sep=3pt] at (B) {};
      \node[dot] at (C) {};
      \node[dot] at (D) {};
      % \ifnum #2 < 90
      %   \draw[] (A) arc[radius = 2, start angle=#2, end angle=90];
      %   \draw[] (B) arc[radius = 2, start angle=90+#2, end angle=180];
      %   \draw[] (C) arc[radius = 2, start angle=180+#2, end angle=270];
      %   \draw[] (D) arc[radius = 2, start angle=270+#2, end angle=360];
      % \fi
    \end{scope}
  }

  \def\eps{0.4}
  \def\rect#1#2#3#4#5{
    \begin{scope}[xshift=-.8cm, yshift=1ex]
      \fill[rect] (#1 - \eps, #2 - \eps)
      -- (#1 - \eps, #4 + \eps) -- (#3 + \eps, #4 + \eps)
      -- (#3 + \eps, #2 - \eps) -- (#1 - \eps, #2 - \eps);
      \node at ($0.5*(#1+#3, #2+#4)$) {}; % {\color{white}\scriptsize\bfseries #5};
    \end{scope}
  }

  \def\rlabel#1#2#3{
    \begin{scope}[xshift=-.8cm, yshift=1ex, every node/.style={}]
      \node[rlabel] at (#1, #2) {\color{white}\scriptsize\bfseries #3};
    \end{scope}
  }

  \def\perm#1#2#3#4#5#6#7#8{
    \begin{scope}[xshift=-.8cm, yshift=1ex, every node/.style={circle, inner sep=0.05cm, draw=red, thick}]
      \node[label=1] at (#1, #2) {};
      \node[label=2] at (#3, #4) {};
      \node[label=3] at (#5, #6) {};
      \node[label=4] at (#7, #8) {};
    \end{scope}
  }

  \centering
  \makebox[\textwidth][c]{
    \begin{tikzpicture}[x=.3cm, y=.3cm, scale=0.8,
        dot/.style={fill=primary, opacity=1.0, circle, inner sep=1.5pt},
        arr/.style={draw, thin, densely dotted},
        bar/.style={line width=4, draw=nuance},
        sync/.style={dotted},
        rect/.style={primary},
        highlight/.style={circle, inner sep=1pt, fill=highlight, opacity=0.8, text opacity=1},
        rlabel/.style={},
      ]

      \node (label) at (-6, 0) {};

      \def\onedeg{0.6}

      \foreach \s in {0,9,...,90} {
          \node (A\s) at (\s * \onedeg, 0) {};
        }

      \node (A85) at (85.5 * \onedeg, 0) {};

      \node (top) at (0, 3) {};
      \node (bot) at (0, 0) {};
      \node (fig) at (0, 12) {};

      \node (arrangements) at (0, -6) {};
      \node (belowarr) at (0, -6) {};
      \node (bbelowarr) at (0, -8) {};
      \node (b4) at (0, -12) {};
      \node (b3) at (0, -13.5) {};
      \node (b2) at (0, -15) {};
      \node (b1) at (0, -16.5) {};

      \motion{A0}{0}
      \motion{A18}{18}
      \motion{A36}{36}
      \motion{A54}{54}
      \motion{A72}{72}
      \motion{A90}{90}

      \draw[thick] (A0 |- b4) node[left] {1} -- (A81 |- b4) .. controls (A85 |- b4) and (A85 |- b3) .. (A90 |- b3) node[right] {1};

      \draw[thick] (A0 |- b3) node[left] {2} -- (A9|- b3) .. controls (A18 |- b3) and (A18 |- b2) .. (A27 |- b2) --  (A45 |- b2) .. controls (A54 |- b2) and (A54 |- b1) .. (A63 |- b1) -- (A90 |- b1) node[right] {2};
      
      \draw[line width=3pt, white] (A0 |- b2) -- (A9 |- b2) .. controls (A18 |- b2) and (A18 |- b3) .. (A36 |- b3) -- (A81 |- b3) .. controls (A85 |- b3) and (A85 |- b4) .. (A90 |- b4);
      \draw[thick] (A0 |- b2) node[left] {4} -- (A9 |- b2) .. controls (A18 |- b2) and (A18 |- b3) .. (A36 |- b3) -- (A81 |- b3) .. controls (A85 |- b3) and (A85 |- b4) .. (A90 |- b4) node[right] {4};
      
      \draw[line width=3pt, white] (A0 |- b1) -- (A45|- b1) .. controls (A54 |- b1) and (A54 |- b2) .. (A63 |- b2) -- (A90 |- b2);
      \draw[thick] (A0 |- b1) node[left] {3} -- (A45|- b1) .. controls (A54 |- b1) and (A54 |- b2) .. (A63 |- b2) -- (A90 |- b2) node[right] {3};

      \begin{scope}[shift=(A0 |- arrangements)]
        \rect{1}{1}{4}{1}{4}
        \rect{1}{4}{4}{4}{2}
        \rect{3}{2}{4}{3}{1}
        \rect{1}{2}{2}{3}{3}

        \rlabel{4}{3}{1}
        \rlabel{3}{4}{2}
        \rlabel{1}{2}{3}
        \rlabel{2}{1}{4}
      \end{scope}

      \begin{scope}[shift=(A9 |- arrangements), rect/.style={secondary}]
        \rect{1}{1}{4}{1}{4}
        \rect{1}{4}{3}{4}{2}
        \rect{4}{2}{4}{3}{1}
        \rect{1}{2}{2}{3}{3}

        \begin{scope}[rlabel/.style={highlight}]
        \rlabel{4}{3}{1}
        \rlabel{3}{4}{2}
        \rlabel{1}{2}{3}
        \rlabel{2}{1}{4}

        \end{scope}
      \end{scope}

      \begin{scope}[shift=(A18 |- arrangements)]
        \rect{1}{1}{4}{1}{4}
        \rect{1}{3}{3}{4}{2}
        \rect{4}{2}{4}{4}{1}
        \rect{1}{2}{3}{2}{3}
        \rlabel{4}{3}{1}
        \rlabel{3}{4}{2}
        \rlabel{1}{2}{3}
        \rlabel{2}{1}{4}

      \end{scope}

      \begin{scope}[shift=(A27 |- arrangements), rect/.style={secondary}]
        \rect{3}{1}{4}{1}{4}
        \rect{1}{3}{3}{4}{2}
        \rect{4}{2}{4}{4}{1}
        \rect{1}{2}{2}{2}{3}

        \begin{scope}[rlabel/.style={highlight}]
          \rlabel{4}{3}{1}
          \rlabel{2}{4}{2}
          \rlabel{1}{2}{3}
          \rlabel{3}{1}{4}
        \end{scope}
      \end{scope}

      \path  (A9 |- belowarr) edge[->, bend right] node[below] {$\sigma_2$} (A27 |- belowarr);

      \begin{scope}[shift=(A36 |- arrangements)]
        \rect{3}{1}{4}{1}{4}
        \rect{1}{3}{3}{4}{2}
        \rect{4}{2}{4}{4}{1}
        \rect{1}{1}{2}{2}{3}

        \rlabel{4}{3}{1}
        \rlabel{2}{4}{2}
        \rlabel{1}{2}{3}
        \rlabel{3}{1}{4}
      \end{scope}

      \begin{scope}[shift=(A45 |- arrangements), rect/.style={secondary}]
        \rect{3}{1}{4}{1}{4}
        \rect{1}{3}{2}{4}{2}
        \rect{3}{2}{4}{4}{1}
        \rect{1}{1}{2}{2}{3}

        \begin{scope}[rlabel/.style={highlight}]
          \rlabel{4}{3}{1}
          \rlabel{2}{4}{2}
          \rlabel{1}{2}{3}
          \rlabel{3}{1}{4}
        \end{scope}
      \end{scope}

      \path  (A27 |- belowarr) edge[->, bend right] node[below] {$1$} (A45 |- belowarr);

      \begin{scope}[shift=(A54 |- arrangements)]
        \rect{3}{1}{4}{2}{4}
        \rect{1}{3}{2}{4}{2}
        \rect{3}{3}{4}{4}{1}
        \rect{1}{1}{2}{2}{3}

        \rlabel{4}{3}{1}
        \rlabel{2}{4}{2}
        \rlabel{1}{2}{3}
        \rlabel{3}{1}{4}

        \begin{scope}[xshift=-.8cm, yshift=1ex, every node/.style={}]
          \draw[{Classical TikZ Rightarrow[length=2]}-{Classical TikZ Rightarrow[length=2]}, very thick, white] (3,2)--(3,3);
        \end{scope}
      \end{scope}
      \path (A54 |- arrangements) -- ++(2,5) node {*};

      \begin{scope}[shift=(A63 |- arrangements), rect/.style={secondary}]
        \rect{3}{1}{4}{3}{4}
        \rect{1}{3}{2}{4}{2}
        \rect{3}{4}{4}{4}{1}
        \rect{1}{1}{2}{2}{3}

        \begin{scope}[rlabel/.style={highlight}]
          \rlabel{4}{4}{1}
          \rlabel{1}{3}{2}
          \rlabel{2}{2}{3}
          \rlabel{3}{1}{4}
        \end{scope}
      \end{scope}

      \path  (A45 |- belowarr) edge[->, bend right] node[below] {$\sigma_1$} (A63 |- belowarr);
      
      \begin{scope}[shift=(A72 |- arrangements)]
        \rect{3}{1}{4}{3}{4}
        \rect{1}{3}{1}{4}{2}
        \rect{2}{4}{4}{4}{1}
        \rect{1}{1}{2}{2}{3}

        \rlabel{4}{4}{1}
        \rlabel{1}{3}{2}
        \rlabel{2}{2}{3}
        \rlabel{3}{1}{4}
      \end{scope}

      \begin{scope}[shift=(A81 |- arrangements), rect/.style={secondary}]
        \begin{scope}[yshift=-.7cm]
          \rect{1}{1}{2}{1}{3}
          \rect{2}{4}{4}{4}{1}
          \rect{3}{2}{4}{3}{4}
          \rect{1}{2}{1}{3}{2}

          \begin{scope}[rlabel/.style={highlight}]
            \rlabel{4}{4}{1}
            \rlabel{1}{3}{2}
            \rlabel{2}{1}{3}
            \rlabel{3}{2}{4}
          \end{scope}
        \end{scope}
      \end{scope}

      \path  (A63 |- belowarr) edge[->, bend right] node[below] {$1$} (A81 |- bbelowarr);

      \begin{scope}[shift=(A90 |- arrangements),rect/.style={fill=nuance}]
        \begin{scope}[yshift=-.7cm]
          \rect{1}{1}{4}{1}{3}
          \rect{1}{4}{4}{4}{1}
          \rect{3}{2}{4}{3}{4}
          \rect{1}{2}{2}{3}{2}

          \begin{scope}[rlabel/.style={highlight}]
            \rlabel{4}{3}{4}
            \rlabel{3}{4}{1}
            \rlabel{1}{2}{2}
            \rlabel{2}{1}{3}
          \end{scope}
        \end{scope}
      \end{scope}

      \path  (A81 |- bbelowarr) edge[->, bend right] node[below] {$\sigma_3$} (A90 |- bbelowarr);
    \end{tikzpicture}}
  \caption{Illustration of Algorithm~\ref{algo:braid-of-cover} for four points moving along a circle, continuing Figure~\ref{fig:example-sep}.
    The arrangement diagrams in blue are the input covering sequence.
    The arrangement diagrams in green represent the intersection of the two neighboring blue diagrams. The pink circled digits highlight a permutation point lying in the arrangement cell.
    The last blue diagram represents an extension of the cover so that the last arrangement is a permutation of the first one (see Section~\ref{subsec:braid-loops}).
  }
  \label{fig:example-braid-of-cover}
\end{figure}
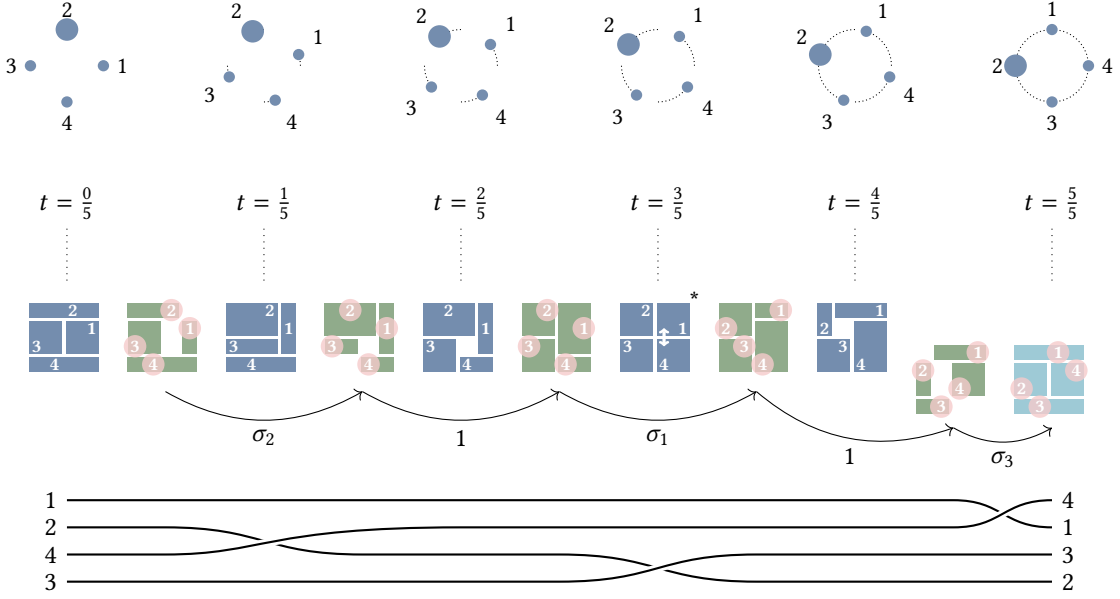

The procedure \emph{braid} (Algorithm~\ref{algo:braid-of-cover}) takes a sequence of arrangements~$A_1,\dotsc,A_r$ as input, with~$A_{i-1} \cap A_{i}$ nonempty.
It computes a sequence of permutation points~$p_1,\dotsc,p_r$ such that~$p_1\in A_1$ and~$p_i \in A_{i-1} \cap A_i$, for~$i\geq 2$.
Simultaneously, it computes the braid~$b$ induced by the piecewise linear path~$p_1 \to \dotsb \to p_r$ as a composition of elementary braids~$B(\pi, \phi)$ (Section~\ref{sec:perm-points-elem}).
Finally, it outputs the triple~$(p_1, b, p_r)$.

Figure~\ref{fig:example-braid-of-cover} shows an example, continuing that of Figure~\ref{fig:example-sep}.
In this case, the algorithm computes (among several possible choices) the following permutation points:
\begin{align*}
  p_1 &= p \left( 4\,3\,1\,2, 3\,4\,2\,1 \right) & p_2 &= p \left( 4\,2\,1\,3, 3\,4\,2\,1 \right)\\
  p_3 &= p_2  & p_4 &= p \left( 4\,1\,2\,3, 4\,3\,2\,1 \right),
\end{align*}
where, for example, $4\,3\,1\,2$ denotes the permutation~$\left\{ 1\mapsto 4, 2\mapsto 3, 3\mapsto 1, 4\mapsto 2 \right\}$;
and outputs the triple~$(p_1, \sigma_2\sigma_1, p_4)$.

% Consider now two permutation points~$p(\pi,\phi)$ and~$p(\pi',\phi')$, both lying in an arrangement cell~$A$.
% The straight-line path $q(\pi) \to p(\pi,\phi) \to p(\pi',\phi') \to q(\pi')$ is entirely contained in~$OC_n$ (because~$A$ is convex),
% and induces a loop in~$C_n$, from and to the base point~$[n]$.
% Therefore, it induces an element of the combinatorial braid group~$B_n$.
% This braid can be computed in terms of a small braid, as follows.

% \begin{lemma}
%   Let~$A$, $\pi$, $\phi$, $\pi'$ and~$\phi'$ as above.
%   The path
%   \[ q(\pi) \to p(\pi,\phi) \to p(\pi',\phi') \to q(\pi') \]
%   induces the braid~$B(\pi \pi'^{-1}, \phi\pi'^{-1})$.
% \end{lemma}

\begin{algorithm}[tbp]
  \begin{description}
    \item[input] A sequence of arrangements.
    \item[output] A triple~$(p, b, p')$ where~$p$ and~$p'$ are permutation points and~$b$ is a braid.
  \end{description}
  \raggedright

  \caption{Computation of a braid associated to an arrangement sequence}
  \begin{pseudo}
    def \fn{braid}(A_1,\dotsc,A_r):\\+
    $(\pi_\mathrm{start},\phi_\mathrm{start}) \gets$ \tn{a permutation point in $A_1$}\\
    $(\pi,\phi) \gets(\pi_\mathrm{start},\phi_\mathrm{start})$\\
    $b \gets 1$ \ct{the identity braid}\\
    for $i$ from $2$ to $r$:\\+
    if $A_{i-1} \cap A_i = \varnothing$:\\+
    error\\-
    $(\pi',\phi') \gets$ \tn{a permutation point in $A_{i-1} \cap A_i$}\label{line:braid-of-cover:perm-point-computation}\\
    $b \gets b \cdot B(\pi \pi'^{-1}, \phi \pi'^{-1})$\\
    $(\pi, \phi) \gets (\pi', \phi')$\\-
    return $p(\pi_\mathrm{start},\phi_\mathrm{start}), b, p(\pi, \phi)$\\-
  \end{pseudo}
  \label{algo:braid-of-cover}
\end{algorithm}

\begin{theorem}[Correctness of \fn{braid}]\label{thm:algo-braid}
  Let~$(A_1,\dotsc,A_r)$ be a sequence of arrangements covering a path~$F : [0,1]\to OC_n$.
  Let $(p, b, p')$ be the output triple of Algorithm~\ref{algo:braid-of-cover}
  on input~$A_1,\dotsc,A_r$.
  Then:
  \begin{itemize}
    \item $F(0)$ and~$p$ are in~$A_1$;
    \item $F(1)$ and~$p'$ are in~$A_r$;
    \item the path~$F(0) \to p \overset{b}{\to} p' \to F(1)$
          lies in~$OC_n$ and is homotopic to~$F$.
  \end{itemize}
\end{theorem}

\begin{proof}
  By definition of a cover (Definition~\ref{def:cover}), the intersections~$A_{i-1} \cap A_i$ are nonempty and therefore each contains a permutation point, by Lemma~\ref{lem:convex-open-cover} and Lemma~\ref{lem:arrangement-permutation-point}.
  So the algorithm does not fail.
  The first two points follow by construction.

  Let~$p_1 = p(\pi_\text{start}, \phi_\text{start})$ and let~$p_i$ be the permutation point chosen in~$A_{i-1}\cap A_i$.
  It is clear that the piecewise linear path~$p_1 \to \dotsb \to p_r$ is also covered by~$A_1,\dotsc,A_r$.
  Its associated braid is~$b$, by repeated application of Lemma~\ref{lem:basic-paths}.
  Since~$F(0) \in A_1$ and~$F(1) \in A_r$, the piecewise linear path
  \[ F(0) \to p_1 \to \dotsb \to p_r \to F(1) \]
  is also covered by the same sequence.
  By Lemma~\ref{lem:homotopy-arrangement}, this path is homotopic to~$F$.
\end{proof}

\subsection{Composition of paths}
\label{sec:composition-paths}

Given two paths~$F, G : [0,1] \to OC_n$ with~$F(1) = G(0)$, the goal is to compute the braid attached to the concatenation~$F \cdot G$ using the algorithms above.
Applying \fn{cover} (Algorithm~\ref{algo:cover}) to~$F$ and~$G$ yields covering sequences of arrangements $(A_1, \dotsc, A_r)$ and~$(B_1, \dotsc, B_s)$.
Their concatenation
\[ A_1, \dotsc, A_r, B_1, \dotsc, B_s \]
covers~$F \cdot G$, so Algorithm~\ref{algo:braid-of-cover} directly produces the desired braid.
However, storing only the outputs of $\fn{braid}(A_1,\dotsc,A_r)$ and~$\fn{braid}(B_1,\dotsc,B_s)$ is preferable; let us denote them
\[ \left( (\pi_F,\phi_F), b_F, (\pi'_F, \phi'_F) \right) \text{ and } \left( (\pi_G,\phi_G), b_G, (\pi'_G, \phi'_G) \right). \]

When an arrangement cell~$A$ contains both $p(\pi'_F, \phi'_F)$ and $p(\pi_G, \phi_G)$ together with the common endpoint~$F(1) = G(0)$, the path composition
\begin{equation}
  \label{eq:1}
  p(\pi_F, \phi_F) \overset{\smash{b_F}}{\to} p(\pi'_F, \phi'_F) \to F(1) = G(0) \to p(\pi_G, \phi_G) \overset{\smash{b_G}}{\to} p(\pi'_G, \phi'_G)
\end{equation}
is homotopic, by convexity of~$A$, to
\[ p(\pi_F, \phi_F) \overset{\smash{b_F}}{\to} p(\pi'_F, \phi'_F) \to p(\pi_G, \phi_G) \overset{\smash{b_G}}{\to} p(\pi'_G, \phi'_G), \]
whose braid equals
\[ b_F \cdot B(\pi'_F \pi_G^{-1}, \phi'_F \pi_G^{-1}) \cdot b_G. \]

The equality~$F(1) = G(0)$ alone, however, does not ensure that $p(\pi'_F, \phi'_F)$ and $p(\pi_G, \phi_G)$ lie in a common arrangement cell; the straight-line segment joining them may even exit~$OC_n$.
We suggest two approaches to circumvent the issue:
\begin{enumerate}
  \item When future compositions are known in advance, append~$B_1$ to~$(A_1,\dotsc,A_r)$.
        The augmented sequence still covers~$F$,
        and this will guarantee that~$F(1)$, $p(\pi'_F, \phi'_F)$ and $p(\pi_G, \phi_G)$ lie in the common cell~$B_1$.
  \item Retain the endpoint cells~$A_r$ and~$B_1$ and compute their intersection, which is nonempty because it contains~$F(1) = G(0)$.
        Choose a permutation point~$p(\pi, \phi)$ in~$A_r \cap B_1$; then the path~\eqref{eq:1} is homotopic to
        \[ p(\pi_F, \phi_F) \overset{\smash{b_F}}{\to} p(\pi'_F, \phi'_F) \to p(\pi, \phi) \to p(\pi_G, \phi_G) \overset{\smash{b_G}}{\to} p(\pi'_G, \phi'_G) \]
        with associated braid
        \[ b_F \cdot B(\pi'_F \pi^{-1}, \phi'_F \pi^{-1}) \cdot B(\pi \pi_G^{-1}, \phi \pi_G^{-1}) \cdot b_G. \]
\end{enumerate}

Next, we relax the condition~$F(1) = G(0)$
and only require equality in~$C_n$, namely, there is a permutation~$\sigma$ such that~$\sigma \cdot F(1) = G(0)$.
We consider the paths~$\widetilde F$ and~$\widetilde G$ in~$C_n$ induced by~$F$ and~$G$ respectively, so that~$\widetilde F(1) = \widetilde G(0)$.
How do we compute the braid induced by the composition~$\widetilde F \cdot \widetilde G$
from the braid data computed for~$F$ and~$G$?
If the permutation~$\sigma$ is known, then the situation reduces to the previous case~$F(1) = G(0)$.
Unfortunately, the separation method \fn{sep} is not enough to compute~$\sigma$.
For example, on two points, assume that~$F(1) = \sigma \cdot G(0)$ for some~$\sigma$ (which is either identity or the transposition~$(1\ 2)$),
and that
\[ \fn{sep}(F, 1, 2, 1) = (2, 1, \re, [\dotsc]) \text{ and }\fn{sep}(G, 1, 2, 0) = (1, 2, \im, [\dotsc]), \]
meaning that~$\re(F_1(1)) > \re(F_2(1))$ and~$\im(G_1(0)) < \im(G_2(0))$.
Is it enough information to recover~$\sigma$? The answer is no:
we could have~$\sigma = \identity$ with~$F(1) = G(0) = (1, \I)$, or~$\sigma = (1\ 2)$ with~$F(1) = \left( 1+\I, 0 \right)$ and~$G(0) = \left( 0, 1+\I \right)$.
Therefore, we need to compute~$\sigma$ by other means, such as path tracking.

\subsection{Braids along loops}
\label{subsec:braid-loops}
Consider a path~$F : [0,1] \to OC_n$
inducing a loop in~$C_n$. In other words, there exists a permutation~$\sigma$ with~$F(0) = \sigma \cdot F(1)$.
We seek a permutation point~$p(\pi, \phi)$ and a braid~$b$ such that
\[ F(0) \to p(\pi, \phi) \overset{b}{\to} \sigma^{-1} \cdot p(\pi, \phi) \to F(1) \]
is homotopic to~$F$.

This problem is similar to the composition problem above.
In particular, we must determine~$\sigma$ by other means, since the \fn{sep} method alone does not provide sufficient information to recover it.

Once~$\sigma$ is known,
we compute a cover~$A_1,\dotsc,A_r$ of~$F$,
which we extend to
\[ A_1,\dotsc,A_r, \sigma^{-1} A_1 \]
Because~$F(0) = \sigma \cdot F(1)$ and~$F(0) \in A_1$, the extended sequence still covers~$F$.
The method $\fn{braid}(A_1,\dotsc,A_r, \sigma^{-1} A_1)$ outputs a triple $(\pi,\phi)$, $b$, and~$(\pi',\phi')$ such that
\[ F(0) \to p(\pi,\phi) \overset{b}{\to} p(\pi',\phi') \to F(1) \]
is homotopic to~$F$.
Moreover, $F(0) = \sigma\cdot F(1) \in A_1$, $p(\pi, \phi) \in A_1$ and~$p(\pi',\phi') \in \sigma^{-1} A_1$.
Hence the path above is homotopic to
\[ F(0) \to p(\pi,\phi) \overset{b}{\to} p(\pi',\phi') \to  \sigma^{-1} \cdot p(\pi, \phi)\to F(1), \]
by convexity of~$\sigma^{-1} A_1$, and it simplifies to
\[ F(0) \to p(\pi, \phi)  \xrightarrow{\tilde b} \sigma^{-1} \cdot p(\pi, \phi) \to F(1), \]
with $\tilde b = b \cdot B(\pi' \sigma^{-1} \pi^{-1}, \phi' \sigma^{-1} \pi^{-1})$, by Lemma~\ref{lem:basic-paths}.

Assume we have already computed \fn{braid}(A_1,\dotsc,A_r) and retained~$A_1$ and~$A_r$.
Analogously to the path composition problem, we can compute a permutation point in the intersection~$A_r \cap \sigma^{-1} A_1$ and use it to close the loop (see Figure~\ref{fig:example-braid-of-cover}).

\section{Implementation}
\label{sec:implementation}
The algorithm presented in the previous sections to compute the braid induced by a path~$F: [0, 1] \to OC_n$ works in two steps: first compute an arrangement sequence covering~$F$ (Algorithm~\ref{algo:cover}); then find a permutation point in the intersection of each pair of consecutive arrangements in the sequence and compute (on the fly) the braid associated to the straight-line interpolation of those permutation points (Algorithm~\ref{algo:braid-of-cover}).

This section revisits Algorithms~\ref{algo:cover} and~\ref{algo:braid-of-cover} from an implementation point of view.
At Line~\ref{line:braid:arrangement-copy} of Algorithm~\ref{algo:cover}, the working arrangement is copied, thus creating a list of arrangements later processed by Algorithm~\ref{algo:braid-of-cover}.
To avoid those copies and simplify implementation, Algorithms~\ref{algo:cover} and~\ref{algo:braid-of-cover} are merged: a current arrangement~$A$ and a permutation point~$p$ in~$A$ are maintained, and whenever~$A$ is updated,~$p$ is updated so that it lies in both the old~$A$ and the new~$A$.
Explicit data structures for arrangements and permutation points are given, and this update is described in detail (Section~\ref{sec:point-pre-arrang}).
Finally, the loop at Line~\ref{line:braid:arrangement-repair} of Algorithm~\ref{algo:cover}, which needs to find incomparable pairs~$(i, j)$ after the deletion of an edge, is revisited (Section~\ref{sec:braid-comp-impl}).
All in all, this leads to simple algorithms that can be implemented with basic data structures.

\subsection{Data structure for pointed arrangements}
Recall that an arrangement is a pair~$(\prec_{\re}, \prec_{\im})$ of partial orders on~$[n]$ such that all~$i \neq j$ are comparable in either~$\prec_{\re}$ or~$\prec_{\im}$ (see Section~\ref{sec:arrangements}).

A partial order~$\prec$ on~$[n]$ is encoded by a directed acyclic graph (DAG)~$G$ with vertices~$[n]$, where~$i \prec j$ holds iff there is a directed path from~$i$ to~$j$ in~$G$.
An arrangement $(\prec_{\re}, \prec_{\im})$ is encoded by a single graph~$G$ whose edges are labelled by~$\re$ or~$\im$, so its edge set is a subset of $[n]^2 \times \{\re, \im\}$.
An output~$(i, j, q, t)$ of~\fn{sep} corresponds to an edge~$(i, j, q)$ with lifetime~$t$.
Let $G_{\re}$ (resp. $G_{\im}$) denote the subgraph of $G$ with edges labelled by $\re$ (resp. $\im$), and write $i \to_{\re} j$ (resp. $i \to_{\im} j$) when $(i, j)$ is an edge of~$G_{\re}$ (resp.~$G_{\im}$).
By metonymy, such a graph representing~$(\prec_{\re}, \prec_{\im})$ is called an \emph{arrangement} if~$(\prec_{\re}, \prec_{\im})$ is an arrangement.

% and $i \to_{\re}^* j$ (resp. $i \to_{\im}^* j$) if there is a directed path of $\re$-labeled (resp. $\im$-labeled) edges from $i$ to $j$. Let $(\prec_{\re}, \prec_{\im})$ be the pair of partial orders associated to $G$.
% Then for all $\reorimsymb \in \{\re, \im\}$, $i \prec_\reorimsymb j \iff i \to_\reorimsymb^* j$. So $G$ is an arrangement if and only if
% \begin{align}\label{eq:graph-arrangement}
%   \forall i, j \in [n],\, \exists \reorimsymb \in \{\re, \im\} \text{ s.t. } i \to_{\reorimsymb}^* j \text{ or } j \to_{\reorimsymb}^* i.
% \end{align}

Let $(\prec_{\re}, \prec_{\im})$ be the pair of partial orders associated to $G$ and let $A$ be the associated cell, as defined in Section~\ref{sec:arrangements}.
Let $\pi, \phi \in S_n$.
A permutation point $p(\pi, \phi)$ is in~$A$
if and only if~$\pi$ increases along the edges of~$G_{\re}$ and~$\phi$ increases along the edges of~$G_{\im}$; that is, for all~$i,j\in [n]$,
\begin{align} \label{eq:pointed-condition}
  i \to_{\re} j \Rightarrow \pi(i) < \pi(j) \quad \text{ and } \quad i \to_{\im} j \Rightarrow \phi(i) < \phi(j).
\end{align}
We write $G \vdash p(\pi, \phi)$ if this condition holds.
In other words, identifying permutations of~$[n]$ with total orders on~$[n]$, $G \vdash p(\pi,\phi)$ if and only if the permutations~$\pi$ and~$\phi$ are topological sorts of~$G_{\re}$ and~$G_{\im}$ respectively.

When computing a braid, we will maintain an arrangement with a permutation point lying in its associated cell; that is a \emph{pointed arrangement}.

% Our data structure to represent a pointed pre-arrangement consists in a pre-arrangement $G$ and four arrays representing two permutations $\pi, \phi$ and their inverses, such that $G \vdash p(\pi, \phi)$.
% In what follows, a pointed pre-arrangement will refer to the mathematical object, as well as to the data structure representing it, depending on the context.

\subsection{Pointed arrangement updates}\label{sec:point-pre-arrang}
The two operations on pointed arrangements are the deletion and the insertion of an edge in $G$, while maintaining the condition $G \vdash p(\pi, \phi)$.
Interpreting the latter condition as $\pi$ and $\phi$ being topological orders on $G_{\re}$ and $G_{\im}$, this is the \emph{dynamic topological sort} problem.
If $G'$ is $G$ with an edge deleted, then $G' \vdash p(\pi, \phi)$ provided that $G \vdash p(\pi, \phi)$, so $\pi$ and $\phi$ can remain unchanged when an edge is deleted from $G$.
On the other hand, inserting an edge in $G$ may require updates to $\pi$ and $\phi$: this is what \fn{insert} does (Algorithm \ref{algo:edge-insert}).
Because $G$ does not change entirely, updating $\pi$ and $\phi$ is preferable to recomputing them from scratch.
The algorithm of \textcite{MarchettispaccamelaNanniRohnert_1996} \parencite[see also][]{PearceKelly_2007} is adapted for this purpose (Figure~\ref{fig:top-sort}).
It modifies $G, \pi$ and $\phi$ \emph{in-place}, so they are global variables in the pseudocode.

\begin{figure}[tbp]
  \centering
  \def\svgwidth{.8\linewidth}
  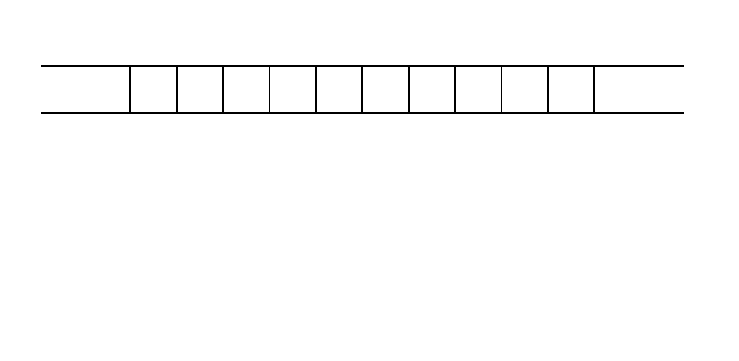
  %\def\svgwidth{.8\linewidth}
  %\import{dyn_topsort}{2.pdf_tex}
  \def\svgwidth{.8\linewidth}
  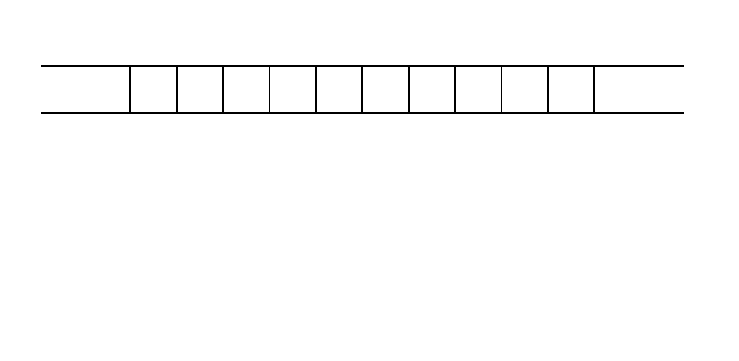
  \caption{Illustration of \fn{update\_real}(i, j).
    We draw $\pi^{-1}$ as a horizontal array mapping indices to vertices (for instance, $\pi(j) = 8$), and draw the $\re$-labeled edges in $G$.
    The topological sort is correct if all edges are pointing to the right.
    The dashed arrow from $i$ to $j$ represents the edge to be inserted.
    Using the variable names from Algorithm \ref{algo:edge-insert}, $m = 8$, $M = 17$ and $\delta$ is the set of vertices corresponding to bold squares.
    Vertex $i$ cannot be in $\delta$, as it is assumed that the new edge does not introduce a cycle.
    The vertices in the interval~$[m, M]$ are reordered, pushing the elements of~$\delta$ to the right and the others to the left.
    Once done, the dashed arrow from $i$ to $j$ is pointing to the right, since~$i$ is not in~$\delta$.
    The procedure always keeps existing edges pointing to the right.
    Although more efficient methods exist for this task  \parencite{PearceKelly_2007}, it is required to update $\pi$ with elementary transpositions to be able to express the braid in standard generators.
  }
  \label{fig:top-sort}
\end{figure}

\begin{algorithm}[tbp]
  \caption{Edge insertion in a pointed arrangement}
  \raggedright
  \begin{pseudo}
    %\begin{noindent}
      def \fn{update\_real}(i, j): \ct{updates $\pi$ after inserting $(i, j, \re)$ in $G$}\\+
        $m \gets \pi(j)$; \quad $M \gets \pi(i)$; \quad $c \gets 1$; \quad $b \gets 1$\\
        \label{line:insert:init}$\delta \gets \{k \in [n]: j \prec_{\re} k \text{ \tn{and} } \pi(k) < M\}$ \ct{Bounded dfs}\\
        for $x$ from $m + 1$ to $M$:\\+
          if $\pi^{-1}(x) \in \delta$:\\+
            $c \gets c + 1$; \quad continue\\-
          for $y$ from $x$ to $x - c + 1$: \ct{Reverse order iteration}\\+
            \label{line:insert:braid-update}if $\phi(\pi^{-1}(y - 1)) < \phi(\pi^{-1}(y))$:\\+
              $b \gets b \cdot \sigma_{y - 1}$\\-
            else:\\+
              $b \gets b \cdot \sigma_{y - 1}^{-1}$\\-
            \label{line:insert:pi-update}$\pi \gets (y - 1\ y) \pi$\\--
        return $b$\\-
      \\
      def \fn{update\_imag}(i, j): \ct{updates $\phi$ after inserting $(i, j, \im)$ in $G$}\\+ 
        \tn{Update $\phi$ with a dynamic topological sorting algorithm}\\
        return $1$\\-
      \\
      
      def \fn{insert}(i, j, \reorimsymb):\\+
        $G \gets G \cup \{(i, j, \reorimsymb)\}$\\
        if $\reorimsymb = \re$:\\+
          return \fn{update\_real}(i, j)\\-
        else:\\+
          return \fn{update\_imag}(i, j)
    %\end{noindent}
  \end{pseudo}
  \label{algo:edge-insert}
\end{algorithm}

% The \fn{insert} function of Algorithm \ref{algo:edge-insert} takes as input an edge $(i, j, q) \in [n]^2 \times \{\re, \im\}$, and updates the pre-arrangement $(G, \pi, \phi)$ by inserting the edge in $G$ and updating $\pi$ and $\phi$ accordingly.
% It also computes the braid associated to the straight-line interpolation of the 
We consider a run of \fn{insert} (Algorithm \ref{algo:edge-insert}).
Let $G', \pi', \phi'$ be the state of the pointed arrangement before the run of \fn{insert}, and let $(i, j, q) \in [n]^2 \times \{\re, \im\}$ be a valid edge for $G'$, meaning that inserting $(i, j, q)$ in $G'$ does not introduce any cycle in $G'_{\re}$ nor $G'_{\im}$.
Let $G, \pi, \phi$ be the state of the arrangement after running $\fn{insert}(i, j, q)$, and $b$ be the output.
\begin{proposition}\label{prop:edge-insert}
  Algorithm \ref{algo:edge-insert} is correct, assuming that the initial graph~$G'$ is an arrangement: with the notation above, $G = G' \cup \{(i, j, q)\}$, $G \vdash p(\pi, \phi)$, and $b$ is the braid induced by the well-defined piecewise straight-line path $p(\pi', \phi') \to p(\pi, \phi)$.
\end{proposition}
\begin{proof}
  When~$q = \im$, the procedure \fn{insert} calls \fn{update\_imag},
  which updates the topological sort~$\phi$ to be compatible with the new edge and returns the identity braid.
  Indeed, the path~$p(\pi', \phi') \to p(\pi, \phi)$ induces the identity braid because~$\pi = \pi'$.

  When~$q = \re$, the procedure \fn{insert} calls \fn{update\_real}, which performs two tasks simultaneously.
  First, it updates the topological sort~$\pi$ to be compatible with the new edge, following the strategy of \textcite{MarchettispaccamelaNanniRohnert_1996} (see Figure~\ref{fig:top-sort}) by left-multiplying~$\pi$ with suitable elementary transpositions (Line~\ref{line:insert:pi-update}).
  Second, it computes the braid induced by~$p(\pi', \phi) \to p(\pi, \phi)$, where~$\pi'$ is the initial value of~$\pi$.

  We claim that, at the end of the procedure, $b$ is the braid induced by~$p(\pi', \phi) \to p(\pi, \phi)$.
  By construction of the topological sort algorithm, $G' \vdash p(\pi,\phi)$ throughout the procedure, so the path $p(\pi', \phi) \to p(\pi, \phi)$ lies in the arrangement cell defined by~$G'$, and Lemma~\ref{lem:basic-paths} applies.
  The path induces the braid $B(\pi' \pi^{-1}, \phi \pi^{-1})$.
  The permutation $\pi' \pi^{-1}$ decomposes as a product of elementary transpositions, accumulated on Line~\ref{line:insert:pi-update}, and the braid is computed using Lemma~\ref{lem:b-braid}: each transposition $(y-1\ y)$ contributes a factor~$\sigma_{y-1}$ or~$\sigma_{y-1}^{-1}$ to~$b$, which is exactly what is computed on Line~\ref{line:insert:braid-update}.
\end{proof}

\begin{figure}[tb]
  \centering
  \def\svgwidth{.6\linewidth}
  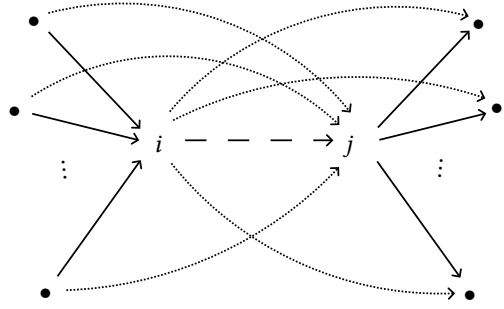
  \caption{Procedure to prepare the deletion of an edge.
    We assume that the edge to be deleted is $i \to_{\re} j$, and we only represent $\re$-labeled edges.
    Initially, only the plain edges and the edge from $i$ to $j$ are present in~$G_{\re}$. After inserting the dotted edges in~$G_{\re}$, inserting an edge~$i\to j$ or~$j \to i$ in~$G_{\im}$, and deleting the edge $i \to_{\re} j$, the graph~$G$ is still an arrangement.
  }
  \label{fig:repair-arrangement}
\end{figure}

\subsection{Braid computation implementation}\label{sec:braid-comp-impl}

\begin{algorithm}[tbp]
  \caption{Braid}
  \raggedright
  \begin{pseudo}
    %\begin{noindent}
      def \fn{braid'}(F):\\+
        $G \gets \{\}$\\
        $\id{lifetime} \gets \left\{  \right\}$ \ct{map edges of~$G$ to their lifetime}\\
        $b \gets 1$\label{line:braid-init-start}\\
        for $i$ in $1, \dotsc, n$:\label{line:graph-init-start}\\+
          for $j$ in $i + 1, \dotsc, n$:\\+
            $e, t \gets \fn{sep}(F, i, j, 0)$\\
            $G \gets G \cup \{e\}$\\
            $\id{lifetime}[e] \gets t$\\--
        $\pi \gets$ \fn{topological\_sort}(G_{\re}); \quad $\phi \gets$ \fn{topological\_sort}(G_{\im})\label{line:braid-init-end}\\
        $G_0, \pi_0, \phi_0 \gets \pi, \phi$ \ct{save initial pointed arrangement} \label{line:graph-init-end}\\
        \\
        loop:\\+
          $(i, j, \reorimsymb) \gets$ \tn{edge of~$G$ with smallest lifetime} \label{line:pop-edge}\\
          if $\id{lifetime}[i,j,q] \geq 1$:\\+
            break\\-
          if \tn{there is no $q$-labeled path from $i$ to~$j$ in $G \setminus \{ (i, j, q) \}$}: \label{line:superfluous-edge}\\+
          $i', j', \reorimsymb', t' \gets \fn{sep}(F, i, j, t)$\label{line:sep}\\
            if $\reorimsymb = \reorimsymb'$: \ct{necessarily, $i=i'$ and $j = j'$} \label{line:q-eq-qp-test}\\+
              $\id{lifetime}[i,j,q] \gets t'$\\
              continue\\-
            else:\label{line:repair-arrangement-prop-start}\\+
              \ct{insert the edge~$(i',j',q')$ in~$G$ and update the topological sorts~$\pi$ and~$\phi$.}\\
              $b \gets \fn{insert}(i', j', \reorimsymb') \cdot b$ \label{line:insert-new}\\
              $\id{lifetime}[i',j',q'] \gets t'$\\
              \ct{insert additional edges to prepare deletion}\\
              for $k \in [n]$ s.t. $k \to_\reorimsymb i$ and not $k \to_q j$: \label{line:repair-start}\\+
                $G \gets G \cup \{(k, j, \reorimsymb)\}$\\
                $\id{lifetime}[k,j,q] \gets t$\\-
              for $k \in [n]$ s.t. $j \to_\reorimsymb k$ and not $i \to_q k$:\\+
                $G \gets G \cup \{(i, k, \reorimsymb)\}$;\\
                $\id{lifetime}[i,j,q] \gets t$\label{line:repair-arrangement-prop-end}\\---
          $G \gets G \setminus \left\{ (i,j,q) \right\}$ \label{line:edge-pop}\\-
          \\
          return $(G_0, \pi_0, \phi_0), b, (G_1, \pi, \phi)$
    %\end{noindent}
  \end{pseudo}
  \label{algo:braid-implementation}
\end{algorithm}

Next, Algorithms \ref{algo:cover} and \ref{algo:braid-of-cover} are merged into the \fn{braid'} function of Algorithm \ref{algo:braid-implementation}.
It relies on Algorithm \ref{algo:edge-insert} to maintain a permutation point within the current arrangement and computes the braid on the fly.

The procedure~\fn{braid'} (Algorithm~\ref{algo:braid-implementation})
first computes an initial arrangement~$G$ containing~$F(0)$ and a permutation point in this arrangement.
Each edge of~$G$ is associated with a lifetime. When all lifetimes are at least~1, the algorithm terminates.

\begin{theorem}[Correctness of \fn{braid'}]\label{thm:algo-braid-bis}
  On input~$F : [0,1] \to OC_n$, given as an approximation,
  Algorithm~\ref{algo:braid-implementation} outputs a triple
  $\left( \left( G_0, \pi_0, \phi_0 \right), b, (G_1, \pi_1, \phi_1) \right)$,
  such that 
  \begin{itemize}
    \item $G_0$ is an arrangement containing~$p(\pi_0, \phi_0)$ and~$F(0)$;
    \item $G_1$ is an arrangement containing~$p(\pi_1, \phi_1)$ and~$F(1)$;
    \item the path~$F(0) \to p(\pi_0, \phi_0) \overset{b}{\to} p(\pi_1,\phi_1) \to F(1)$
          lies in~$OC_n$ and is homotopic to~$F$.
  \end{itemize}
\end{theorem}

\subsubsection{Basic invariant}\label{sec:basic-invariant}

We show that the main loop maintains the following invariants: $G$ is an arrangement and $G \vdash p(\pi, \phi)$.

At each iteration of the main loop, an edge in~$G$ with the smallest lifetime is picked for deletion, say~$(i, j, q)$.
Deleting the edge may break the invariant that~$G$ is an arrangement by creating incomparable pairs.
If there is a $q$-labeled path from~$i$ to~$j$ in~$G\setminus \{ (i,j, q) \}$, then deleting the edge does not create incomparable pairs.
Otherwise, the procedure inserts new edges to make the deletion possible (the block following Line~\ref{line:superfluous-edge}).

To this end, the procedure queries~\fn{sep} (Line~\ref{line:sep}), which outputs a new edge~$(i',j',q')$ with a lifetime~$t'$.
If~$q = q'$ (Line~\ref{line:q-eq-qp-test}), then $i = i'$ and~$j = j'$, since otherwise~$i = j'$ and~$j = i'$, which would mean that both~$q(F_i(t)) < q(F_j(t))$ and~$q(F_i(t)) > q(F_j(t))$ hold.
So in this case, the procedure does not delete the edge but extends its lifetime to~$t'$, and proceeds to the next iteration.

In the case where~$q \neq q'$ (Line~\ref{line:repair-arrangement-prop-start}), the new edge is inserted, using the \fn{insert} procedure to keep~$\pi$ and~$\phi$ updated and to compute the braid on the fly.
But the expired edge~$(i,j,q)$ is not yet ready for deletion: $G \setminus \left\{ (i,j,q) \right\}$ may still contain incomparable pairs.
Assume that there is indeed one such pair~$\{ u, v \}$, which means that there is a $\re$- or $\im$-labeled path from~$u$ to~$v$ (or the other way) in~$G$ but not in~$G \setminus \left\{ (i,j, q) \right\}$.
This cannot be $\{ i, j \}$ since~$i$ and~$j$ are comparable thanks to the newly inserted edge~$(i',j',q')$.
To fix this, assume that there is a $\re$- or $\im$-labeled path from~$u$ to~$v$ in~$G$ and~$v \neq j$ (the case~$u\neq i$ is similar). This path does not exist in~$G \setminus \left\{ (i,j, q) \right\}$, so it is $q$-labeled and goes through the edge~$i \to_q j$.
Since~$j \neq v$, there is some edge~$j \to_q k$ following~$i \to_q j$ in the path from~$u$ to~$v$.
The procedure precisely inserts all the edges $i \to_q k$ (whenever~$j\to_q k$) and~$k \to_q i$ (whenever~$k\to_q j$), see Figure~\ref{fig:repair-arrangement}.
Therefore, after Line~\ref{line:repair-arrangement-prop-end}, the segment~$i\to_q j \to_q k$ in the path from~$u$ to~$v$ can be replaced by the edge~$i\to_q k$, and we obtain a path from~$u$ to~$b$ in~$G \setminus \left\{ (i, j, q) \right\}$.
Therefore, there are no incomparable pairs in~$G\setminus \left\{ (i, j, q) \right\}$.
Note that the edges $i \to_q k$ and~$k \to_q i$ can be safely inserted, with lifetime~$t$, because they do not change the relation~$\prec_q$.

\subsubsection{Termination}

Let~$\varepsilon > 0$ be given by Condition~\ref{point:approx-2}.
At each iteration, the loop selects an edge with the smallest lifetime~$t$.
When \fn{sep} is called on this edge, the returned lifetime satisfies $t' \geq \min(1, t + \varepsilon)$, so each call increases the lifetime of the corresponding relation by at least~$\varepsilon$ until it reaches~$1$.
Since there are finitely many ordered pairs and labels, this can only happen finitely many times.

When inserting additional edges without lifetime increase (Lines~\ref{line:repair-arrangement-prop-start}--\ref{line:repair-arrangement-prop-end}), the number of paths in~$G_q$ going through the new edges $(i, k, q)$ or~$(k, j, q)$ is strictly less than the number of paths going through~$(i, j, q)$, which is deleted. This ensures that the main loop does not stall at a given~$t$.

\subsubsection{Correction}

It remains to check that~\fn{braid'} computes what we want.
Let~$G_k$, $\pi_k$, $\phi_k$ be the values of the corresponding variables at the end of the $k$th iteration.
Considering the state of the variables just after Line~\ref{line:edge-pop}, the basic invariant (Section~\ref{sec:basic-invariant}) tells us that~$G_k$ contains~$p(\pi_k, \phi_k)$. Considering the state just before Line~\ref{line:edge-pop}, it tells us that~$G_{k-1}$ also contains~$p(\pi_{k}, \phi_{k})$.
Moreover, since new edges are inserted following~\fn{sep}, or using transitivity of partial orders (Lines~\ref{line:repair-start}--\ref{line:repair-arrangement-prop-end}), the arrangement sequence~$(G_k)$ covers~$F$.

By Proposition~\ref{prop:edge-insert}, at the~$k$th iteration, \fn{insert}(i',j',q') returns the braid induced by the path~$p(\pi_{k-1}, \phi_{k-1}) \to p(\pi_k, \phi_k)$. The iterations where~\fn{insert} is not called do not modify~$\pi$ and~$\phi$.
Then, correctness follows in the same way as that of Algorithm~\ref{algo:braid-of-cover} (Theorem~\ref{thm:algo-braid}).

\subsubsection{Arrangement initialization}
Let us comment on the computation of an initial arrangement~$G$ (Lines \ref{line:graph-init-start}-\ref{line:graph-init-end}).
For each distinct $i, j \in [n]$, we compute an edge between $i$ and $j$ using \fn{sep}, resulting in a quadratic cost in $n$.
This ensures that $G$ is an arrangement, though there may be redundant edges.
Depending on the specific implementation of \fn{sep} and with additional conditions on $F(0)$, this initialization can be improved.
For example, if \fn{sep} only returns $\re$-labeled edges at time $0$, a way to initialize the arrangement is to sort the components of $F(0)$ by real part using \fn{sep}.
The arrangement $G$ for which $G_{\im}$ is empty and $G_{\re}$ contains edges $(i, i + 1, \re)$ for $i \in [n - 1]$ will contain $F(0)$.
This initialization process is quasi-linear in $n$, as opposed to the quadratic cost of the naive method.
This strategy can be adapted to the case where \fn{sep} only returns $\re$-labeled edges most of the time.
Moreover, we may use a priori information on $F(0)$ to directly compute an arrangement containing it, such as the components of $F(0)$ being lexicographically sorted.

\def\motionbis#1#2{
  \begin{tikzpicture}[x=.3cm, y=.3cm, scale=0.8,
      dot/.style={fill=primary, opacity=1.0, circle, inner sep=1.5pt},
      arr/.style={draw, thin, densely dotted},
      bar/.style={line width=4, draw=nuance},
      sync/.style={dotted},
      rect/.style={primary},
      highlight/.style={circle, inner sep=1pt, fill=highlight, opacity=0.8, text opacity=1},
      rlabel/.style={},
      baseline=(center.base)
    ]
    \tikzmath{integer \iter; \iter=#2/18;}
    \begin{scope}[shift=(#1 |- fig), label distance=-2, radius = 2]
      \path[arr] (2,0)  arc[start angle=0, end angle=#2]  node[label={#2:\small 1}] (A) {};
      \path[arr] (0,2)  arc[start angle=90, end angle=90 + #2]  node[inner sep=3pt,  label={90+#2:\small 2}] (B) {};
      \path[arr] (-2,0)  arc[start angle=180, end angle=180+#2]  node[label={180+#2:\small 3}] (C) {};
      \path[arr] (0,-2)  arc[start angle=270, end angle=270+#2]  node[label={270+#2:\small 4}] (D) {};

      \node (center) at (0, 0) {};
      \node[dot] at (A) {};
      \node[dot, inner sep=3pt] at (B) {};
      \node[dot] at (C) {};
      \node[dot] at (D) {};
    \end{scope}
  \end{tikzpicture}
}

\begin{figure}
  \centering
  \begin{tabular}{cccccc}
    \toprule
    $t$        & $G$                                                                                          & $\pi^{-1}$     & $\phi^{-1}$    & $b$                 & $F(t)$ \\
    \midrule
    $0$        & \begin{tikzcd}
                   1 \arrow[r, dotted]                                      & 2                              \\
                   4 \arrow[u, dotted] \arrow[r, dotted] \arrow[ru, dotted] & 3 \arrow[u, dotted] \arrow[lu]
                 \end{tikzcd} & $3\ 4\ 2\ 1\ $ & $4\ 3\ 1\ 2\ $ & $1$                 &
    \motionbis{A0}{0}                                                                                                                                                       \\
    $\frac 15$ & \begin{tikzcd}
                   1                                                        & 2 \arrow[l]                    \\
                   4 \arrow[u, dotted] \arrow[r, dotted] \arrow[ru, dotted] & 3 \arrow[u, dotted] \arrow[lu]
                 \end{tikzcd} & $3\ 4\ 2\ 1\ $ & $4\ 3\ 1\ 2\ $ & $1$                 &
    \motionbis{A18}{18}                                                                                                                                                     \\
    $\frac 25$ & \begin{tikzcd}
                   1                                      & 2 \arrow[l]                              \\
                   4 \arrow[u, dotted] \arrow[ru, dotted] & 3 \arrow[u, dotted] \arrow[l] \arrow[lu]
                 \end{tikzcd}         & $3\ 4\ 2\ 1\ $ & $4\ 3\ 1\ 2\ $ & $1$                 &
    \motionbis{A36}{36}
    \\
    $\frac 35$ & \begin{tikzcd}
                   1                   & 2 \arrow[l] \arrow[ld]                   \\
                   4 \arrow[u, dotted] & 3 \arrow[u, dotted] \arrow[l] \arrow[lu]
                 \end{tikzcd}                            & $3\ 2\ 4\ 1\ $ & $4\ 3\ 1\ 2\ $ & $\sigma_2$          & \motionbis{A54}{54}                                      \\
    $\frac 45$ & \begin{tikzcd}
                   1                   & 2 \arrow[l] \arrow[ld]                           \\
                   4 \arrow[u, dotted] & 3 \arrow[u, dotted] \arrow[l] \arrow[lu, dotted]
                 \end{tikzcd}                    & $3\ 2\ 4\ 1\ $ & $4\ 3\ 1\ 2\ $ & $\sigma_2$          & \motionbis{A72}{72}                                              \\
    \midrule
    $1$        &
    \begin{tikzcd}[ampersand replacement=\&]
      1                   \& 2 \arrow[l, dotted] \arrow[ld]                           \\
      4 \arrow[u, dotted] \& 3 \arrow[u, dotted] \arrow[l, dotted] \arrow[lu, dotted]
    \end{tikzcd}
               & $3\ 2\ 4\ 1\ $                                                                               & $3\ 4\ 2\ 1\ $ & $\sigma_2$     & \motionbis{A90}{90}       \\
               && $2\ 3\ 1\ 4\ $ & $3\ 2\ 4\ 1\ $ & $\sigma_2 \sigma_1 \sigma_3$ & \\
  \end{tabular}
  \caption{Braid computation of the path described by Figure \ref{fig:example-sep}.
    $t$: time in $[0, 1]$.
    $G$: arrangement represented by a graph. Plain edges correspond to $\re$-labeled edges, dashed edges correspond to $\im$-labeled edges.
    $\pi^{-1}, \phi^{-1}$: permutations on $[4]$ that represent a permutation point, given by the array of their evaluations.
    $b$: braid on $4$ strands.
    $F$: state of the path $F$ at time $t$.
    The upper half of the table records the values of those variables in Algorithm \ref{algo:braid-implementation} after initialization and at the end of each iteration.
    The column $F$ suggests that $F(0) = \sigma \cdot F(1)$ with $\sigma = 2\ 3\ 4\ 1\ $.
    In the lower half, we close the loop by adding $\sigma^{-1} G_0$ to the arrangement sequence, as described in \ref{subsec:braid-loops}.
    Denote by $\pi_0, \phi_0$ (resp. $\pi_1, \phi_1$) the permutations at $t = 0$ (resp. $t = 1$). To wrap up the computation, we multiply $b$ by $\sigma_1 \sigma_3$, the braid induced by the straight-line path $p(\pi_1, \phi_1) \to \sigma^{-1} \cdot p(\pi_0, \phi_0)$ which can be computed using Lemma \ref{lem:b-braid}.
    The output $\sigma_2 \sigma_1 \sigma_3$ is the braid induced by $F$,
    in the sense of Theorem~\ref{thm:algo-braid-bis}.
  }
  \label{fig:algo-run}
\end{figure} \noindent

\subsection{Practical implementation of \emph{sep}}
Our main motivation is to compute the braid defined by the displacement of the roots of a parametrized polynomial $f \in \C[t][x]$, when $t$ moves along a continuous path $\gamma : [0, 1] \to \C$ that stays away from the discriminant locus of $f$. To compute a braid using the methods presented above, we need an approximation of the path in~$OC_n$ formed by the roots of~$f(\gamma(t), -)$, in the sense of Definition~\ref{def:approx2}.

A certified homotopy continuation software (such as \cite{BeltranLeykin_2013,GuillemotLairez_2024,DuffLee_2024})
takes as input $f$ and $\gamma$, and returns disjoint tubular neighborhoods around the root paths.
To be more specific, we consider \emph{algpath}\footnote{\url{https://gitlab.inria.fr/numag/algpath}} \parencite{GuillemotLairez_2024} which relies on interval arithmetic.

Let $F: [0, 1] \to OC_n$ be a path describing the roots of $f(\gamma(t), -)$ as $t$ varies.
Denote by $\square [0, 1]$ the closed subintervals of $[0, 1]$, and by $\square \C$ the set of pairs of real compact intervals (that is, boxes in the plane). 
For each $i \in [n]$, \emph{algpath} returns an interval extension $\square F_i: \square [0, 1] \to \square \C$ of $F_i$. They are represented by sequences of Taylor models \parencite[Section~6.2]{GuillemotLairez_2024}. We can ensure that for all $t \in [0, 1]$ and distinct $i, j \in [n]$,
\begin{enumerate}
  \item $\square F_i(\{t\}) \cap \square F_j(\{t\}) = \varnothing$;\label{output:1}
  \item $\square F_i$ is continuous at $\{t\}$ for the Hausdorff topology.\label{output:2}
\end{enumerate}
% We do not require $\square F_i(\{t\})$ to be tight around $F_i(t)$; it can be a large box around it as long as it isolates $F_i(t)$, in the sense of Condition \ref{output:1} on the output.
The procedure \fn{sep} is implemented as follows. Let $i, j \in [n]$ distinct and $t \in [0, 1)$. Set $\delta = 1 - t$, then
\begin{enumerate}
  \item compute $X_i = \square F_i([t, t + \delta]) \subseteq \C$ and $X_j = \square F_j([t, t + \delta]) \subseteq \C$;
  \item \label{algo:sep-step2}
        \begin{itemize}
          \item if $\re X_i < \re X_j$, return $(i, j, \re, t + \delta)$;\\
          \item if $\re X_i > \re X_j$, return $(j, i, \re, t + \delta)$;\\
          \item if $\im X_i < \im X_j$, return $(i, j, \im, t + \delta)$;\\
          \item if $\im X_i > \im X_j$, return $(j, i, \im, t + \delta)$;
        \end{itemize}
  \item else divide $\delta$ by $2$ and go back to step $1$.
\end{enumerate}
Since $\square F_i(\{t\})$ and $\square F_j(\{t\})$ are disjoint, continuity of $\square F_i$ and $\square F_j$ at $\{t\}$ ensures that for small enough $\delta$, $\square F_i([t, t + \delta])$ and $\square F_j([t, t + \delta])$ are disjoint.
The latter being complex intervals, one of the conditions of step \ref{algo:sep-step2} must hold, so the procedure terminates.
As for Definition \ref{def:approx}, Items \ref{point:approx-1} and \ref{point:approx-2} are straightforward, while Item \ref{point:approx-3} follows from compactness of $[0, 1]$ and Conditions \ref{output:1} and \ref{output:2}.

A Rust implementation of the presented braid computation algorithms is available at
\begin{quote}
  \url{https://gitlab.inria.fr/aguillem/braids}
\end{quote}
It provides a command-line interface to compute braid monodromy, with \emph{algpath} as the certified homotopy continuation backend, for which \fn{sep} has been implemented as explained above.

\renewcommand*{\bibfont}{\small}
\printbibliography

\end{document}